\documentclass[journal]{IEEEtran}
\usepackage{cite}
\usepackage{amsmath,amssymb,amsfonts}
\usepackage{algorithmic}
\usepackage[linesnumbered,ruled,vlined]{algorithm2e}
\usepackage{graphicx}
\usepackage{textcomp}
\usepackage{xcolor}
\usepackage{enumerate}
\usepackage{booktabs}
\usepackage{diagbox}
\usepackage{supertabular}
\usepackage{caption}
\usepackage{float} 
\usepackage{subcaption}
\usepackage{verbatim}
\usepackage{bm}
\usepackage{breqn}
\usepackage{stfloats}
\usepackage{amsthm}
\usepackage{hyperref}

\newtheorem{theorem}{Theorem}

\newtheorem{problem}{Problem}

\newtheorem{proposition}{Proposition}

\begin{document}

\title{Optimal AoI-based Block Propagation and Incentive Mechanism for Blockchain Networks in Web 3.0}
	
\author{Jinbo Wen, Jiawen Kang*,  Zehui Xiong, Hongyang Du, Zhaohui Yang,  Dusit Niyato, \textit{Fellow, IEEE}, \\  Meng Shen, Yutao Jiao, and Yang Zhang

\thanks{ 
	    J. Wen and Y. Zhang are with the College of Computer Science and Technology, Nanjing University of Aeronautics and Astronautics, China (e-mails: jinbo1608@nuaa.edu.cn; yangzhang@nuaa.edu.cn). 
        J. Kang is with the School of Automation, Guangdong University of Technology, China (e-mail: kavinkang@gdut.edu.cn). 
        Z. Xiong is with the Pillar of Information Systems Technology and Design, Singapore University of Technology and Design, Singapore (e-mail: zehui\_xiong@sutd.edu.sg). 
        H. Du is with the Department of Electrical and Electronic Engineering, University of Hong Kong, Hong Kong (e-mail: duhy@eee.hku.hk).
        Z. Yang is with the College of Information Science and Electronic Engineering, Zhejiang University, China (e-mail: yang\_zhaohui@zju.edu.cn).
        D. Niyato is with the School of Computer Science and Engineering, Nanyang Technological University, Singapore (e-mail: dniyato@ntu.edu.sg).
        M. Shen is with the School of Cyberspace Science and Technology, Beijing Institute of Technology, China (e-mail: shenmeng@bit.edu.cn). 
        Y. Jiao is with the College of Communications Engineering, Army Engineering University of PLA, China (e-mail: yjiao001@yeah.net).  
        The work was presented in part at the 21st IEEE International Conference on Trust, Security and Privacy in Computing and Communications (\textit{*Corresponding author: Jiawen Kang}).
	} 
}
%
	
\maketitle
	
\begin{abstract}
Web 3.0 is regarded as a revolutionary paradigm that enables users to securely manage data without a centralized authority. Blockchains, which enable data to be managed in a decentralized and transparent manner, are key technologies for achieving Web 3.0 goals. However, Web 3.0 based on blockchains is still in its infancy, such as ensuring block freshness and optimizing block propagation for improving blockchain performance. In this paper, we develop a freshness-aware block propagation optimization framework for Web 3.0. We first propose a novel metric called Age of Block Information (AoBI) based on the concept of age of information to quantify block freshness. AoBI measures the time elapsed from the freshest transaction generation to the completion of block consensus. To make block propagation optimization tractable, we classify miners into five different states and propose a block propagation model for public blockchains inspired by epidemic models. Moreover, considering that the miners are bounded rational, we propose an incentive mechanism based on the evolutionary game for block propagation to improve block propagation efficiency. Numerical results demonstrate that compared with other block propagation mechanisms in public blockchains, the proposed scheme has a higher block forwarding probability, which improves block propagation efficiency and decreases the minimum value of average AoBI.
\end{abstract}

\IEEEpeerreviewmaketitle

\section{Introduction}

With the advancement of cutting-edge technologies such as blockchain technologies, Web 3.0 has gained significant attention because of its unique decentralized characteristics\cite{chen2022digital, 10542397}. {Since the creation of the World Wide Web, there have been three generations of the Web. Web 1.0 was created to build information networks, which are characterized by providing users with static information and reading through centralized architectures\cite{chen2022digital}. Web 2.0 is a paradigm shift in how the internet is used, which is characterized by interactivity and social connectivity\cite{chen2022digital}.} Nowadays, Web 3.0 is emerging as the next potential generation of information infrastructures, which is described as a decentralized Internet\cite{lin2023unified}. Different from the focus of Web 2.0 on establishing user interaction with content on the Internet, Web 3.0 focuses on users' control of their own generated content based on blockchain technologies and decentralized wireless edge computing architectures\cite{lin2023unified}.

Blockchains as distributed ledger technologies have attracted widespread attention from both academia and industry\cite{10528325}. Based on encryption technologies and consensus algorithms of distributed systems, blockchains can effectively solve security vulnerabilities caused by centralized nodes and the problem of the single point of failure\cite{wenoptimal, 10528325}. {Since blockchains can achieve cross-domain trust in the highly distributed system without a trusted center, they play an important role in many fields, such as smart cities\cite{xu2023quantum}, Internet of Vehicles\cite{shen2022secure}, and metaverses\cite{kang2022blockchain}.} In Web 3.0, there are plenty of transactions of user-generated data and digital products (e.g., Non-Fungible Tokens (NFTs)) between users, which requires the use of blockchain technologies for the secure storage and efficient management of transactions\cite{chen2022digital}. {Therefore, blockchains are regarded as indispensable and core technologies for Web 3.0\cite{xu2023quantum}.}

{Although Web 3.0 is more accessible, efficient, and intelligent than previous generations, it still faces many challenges, such as improving blockchain performance\cite{xu2023quantum, lin2023unified}.} Especially, in public blockchains, a new block is broadcast randomly to most miners (or even all miners) in the miner network for validation, which causes large overall propagation time\cite{nakamoto2008bitcoin}. {When block propagation time in the network is too long, it may lead to insufficient signature collections or excessive numbers of forks\cite{9357330}. Besides, too large propagation delay may significantly prolong the generation interval of blocks, which results in poor block freshness. Therefore, to effectively improve blockchain performance, optimizing block propagation is critical\cite{zhang2022speeding, 10542397}. Some efforts have been conducted to optimize block propagation \cite{9917546, sallal2022security, li2021enhancing, ersoy2018transaction}, but they do not consider block freshness.} In the literature, Age of Information (AoI) is a well-accepted metric to quantify data freshness, but it ignores the data processing procedure\cite{ying2022aoti}. Recent studies like Age of Processing (AoP)\cite{li2021age} and Age of Task (AoT)\cite{8902529} improve the AoI by taking data processing time into account. However, they cannot quantify block propagation delay based on the random propagation of public blockchains.

To address the above challenges, in this paper, we first propose a novel metric called Age of Block Information (AoBI) based on the concept of AoI to measure block freshness in public blockchains, {where block freshness represents the timeliness of the information contained within the block, indicating the degree to which the block reflects the most up-to-date state of the distributed ledger. Second, since a new block spreads across the miner network in the form of \textit{rumor mongering} \cite{kan2018boost}, the dynamic behavior of miners during the new block spreading can be captured by epidemic models\cite{jiang2023approaching}. Based on the social theory that is used to study and reflect social phenomena, we propose a block propagation model for public blockchains inspired by epidemic models, which makes the block propagation optimization tractable. Furthermore, considering that the block propagation process is a dynamic scenario, and miners are non-cooperative when they propagate the block, we formulate an incentive mechanism based on the evolutionary game theory for miners to propagate a new block rationally rather than randomly, thus optimizing block propagation and ensuring block freshness. Based on extensive simulations, we discover factors that affect blockchain performance with security. The main contributions of this paper are summarized as follows:}
\begin{itemize}
    \item {To measure block freshness for public blockchains, we design a novel metric called AoBI based on the concept of AoI, which considers the procedures of block processing, block validation, and block propagation.}
    \item {To make the block propagation optimization tractable, we innovatively classify miners into five different states according to different behaviors during miners propagating the block. We then propose a block propagation model inspired by epidemic models for public blockchains.}
    \item {To achieve block propagation optimization, we formulate an incentive mechanism based on the evolutionary game from the perspective of block validation and block propagation, which considers the rationality of miners by analyzing their behaviors dynamically.}
    \item To highlight the improvement of block propagation, we conduct extensive simulations on the proposed incentive mechanism. By comparing with other block propagation mechanisms, the block forwarding probability of the proposed incentive mechanism is higher and reaches the upper bound faster, which demonstrates the efficiency of the proposed scheme.
\end{itemize}

The remainder of the paper is organized as follows: In Section \ref{relate}, we review the related work. In Section \ref{framework}, we propose a freshness-aware block propagation optimization framework for Web 3.0. In Section \ref{AoBI_minmization}, we formulate the average AoBI minimization problem. Section \ref{optimize} presents the block propagation model for public blockchains and the incentive mechanism based on the evolutionary game for optimizing block propagation. Section \ref{result} presents the experiment results. Finally, Section \ref{conclude} concludes this paper.

\section{Related Work}\label{relate}
\subsection{Blockchain-enabled Web 3.0}
Web 3.0, also known as the semantic web, is the next frontier in web development based on Artificial Intelligence (AI), the Internet of Things (IoT), and blockchain technologies. 
Nowadays, blockchain-enabled Web 3.0 has attracted significant attention from both academics and industry, and some efforts have been conducted to achieve blockchain-enabled Web 3.0\cite{liu2021make,ragnedda2019blockchain,xu2023quantum,lin2023unified}. For example, Ragnedda \textit{et al.} \cite{ragnedda2019blockchain} discussed how the advent of blockchain technologies brings the third era of the web, i.e., Web 3.0. Xu \textit{et al.} \cite{xu2023quantum} proposed a quantum blockchain-driven Web 3.0 framework that comprises an enabling infrastructure, quantum cryptography protocols, and quantum blockchain-based services. Furthermore, the authors explored potential challenges and applications of implementing quantum blockchain in Web 3.0 \cite{xu2023quantum}. Lin \textit{et al.} \cite{lin2023unified} proposed an integrated framework connecting semantic ecosystems and blockchain for wireless edge-intelligence enabled Web 3.0, which can avoid information overloading to users. Moreover, the authors proposed an adaptive Deep Reinforcement Learning (DRL)-based sharding mechanism to improve interaction efficiency, thus improving the performance of Web 3.0 services\cite{lin2023unified}.{However, most existing works do not consider the optimization of blockchain performance to fundamentally improve Web 3.0 performance. Therefore, it is necessary to optimize blockchain performance to enable Web 3.0, especially in optimizing block propagation.}

\subsection{Data Freshness Metrics}
{As a well-established metric, AoI is defined as the elapsed time from the generation of the latest received status update\cite{yates2021age}, which can effectively quantify data freshness at the destination\cite{Jinbo}. AoI has been widely used in many applications, such as federated learning\cite{kang2022blockchain}, AI-generated content networks\cite{wen2023freshness}, and metaverses\cite{Jinbo,kang2022blockchain}. However, AoI completely ignores the data processing procedure.} Therefore, some works take the data processing procedure into account to improve the AoI\cite{ying2022aoti,li2021age,8902529}. Ying \textit{et al.}\cite{ying2022aoti} proposed a novel metric called Age of Task-oriented Information (AoTI) to measure the freshness of industrial tasks in industrial wireless sensor networks. Li \textit{et al.}\cite{li2021age} proposed a new metric called AoP to quantify the freshness of status data in intelligent IoT applications, such as video surveillance, which is defined as the time elapsed since the newest received processed status data is generated. Song \textit{et al.}\cite{8902529} proposed a performance metric called AoT to evaluate the temporal value of computation tasks, which is defined as the time elapsed since the first unprocessed task left in the queue is generated.  {However, the existing metrics cannot quantify the block propagation delay of public blockchains due to the characteristic of random propagation. Motivated by the above works, we aim to design a new metric to measure block freshness for public blockchains.} 

\subsection{Incentive Mechanisms for Blockchain Networks}
The integration of blockchain and incentive mechanisms is a hot topic for blockchain enhancement\cite{jiao2019auction, wang2023connectivity,li2021contract}. Jiao \textit{et al.} \cite{jiao2019auction} proposed an auction-based market model to allocate computing resources in blockchain networks. Wang \textit{et al.} \cite{wang2023connectivity} proposed a multidimensional contract to incentive IoT devices to join the construction of the wireless blockchain network. Li \textit{et al.} \cite{li2021contract} proposed two joint models under the contract theory to bridge blockchain and IoT users, which balances the security incentive and economic incentive. However, the above works do not tackle the problem of incentive mechanism design for optimizing block propagation to enhance the performance of the blockchain system.

Research on optimizing block propagation for public blockchains can be divided into three directions: 1) \textit{Optimizing blockchain network topology}\cite{9917546,sallal2022security}; 2) \textit{Optimizing block verification}\cite{decker2013information,li2021enhancing}; 3) \textit{Optimizing the propagation behavior of miners}\cite{ersoy2018transaction}. 
From the perspective of optimizing blockchain network topology, Wang \textit{et al.}\cite{9917546} proposed a broadcasting mechanism that optimizes the blockchain network topology and broadcasts the transmission process based on unsupervised learning and greedy algorithms, thus reducing the propagation latency of the blockchain network. Sallal \textit{et al.} \cite{sallal2022security} proposed a clustering protocol that divides a blockchain network into several clusters and selects a master miner for every cluster, thus increasing blockchain network connectivity and decreasing block propagation delay. From the perspective of optimizing block verification, Li \textit{et al.} \cite{li2021enhancing} proposed a probabilistic verification scheme to reduce block propagation delay, where each miner can choose whether to verify the new block based on a probability. Decker \textit{et al.}\cite{decker2013information} proposed a protocol that minimizes block verification and pipelines block propagation, thereby reducing block propagation delay. From the perspective of optimizing the propagation behavior of miners, few works have been conducted on incentive mechanism design for optimizing the propagation behavior of miners. Ersoy \textit{et al.} \cite{ersoy2018transaction} proposed a propagation mechanism to encourage miners to propagate messages and a routing mechanism to reduce the redundant communication cost.

{However, most existing works do not take miner rationality and block freshness into account when optimizing block propagation. In Web 3.0, rational users can freely engage and collaboratively manage this ecosystem\cite{liu2021make}. Therefore, it is still challenging to optimize block propagation by considering miner rationality and block freshness.
Motivated by the aforementioned research gaps, we propose a freshness-aware block propagation optimization framework for Web 3.0.} 


\section{Freshness-aware Block Propagation Optimization Framework for Web 3.0}\label{framework}
\begin{figure*}[t]
\vspace{-0.5cm}
\centerline{\includegraphics[width=0.88\textwidth]{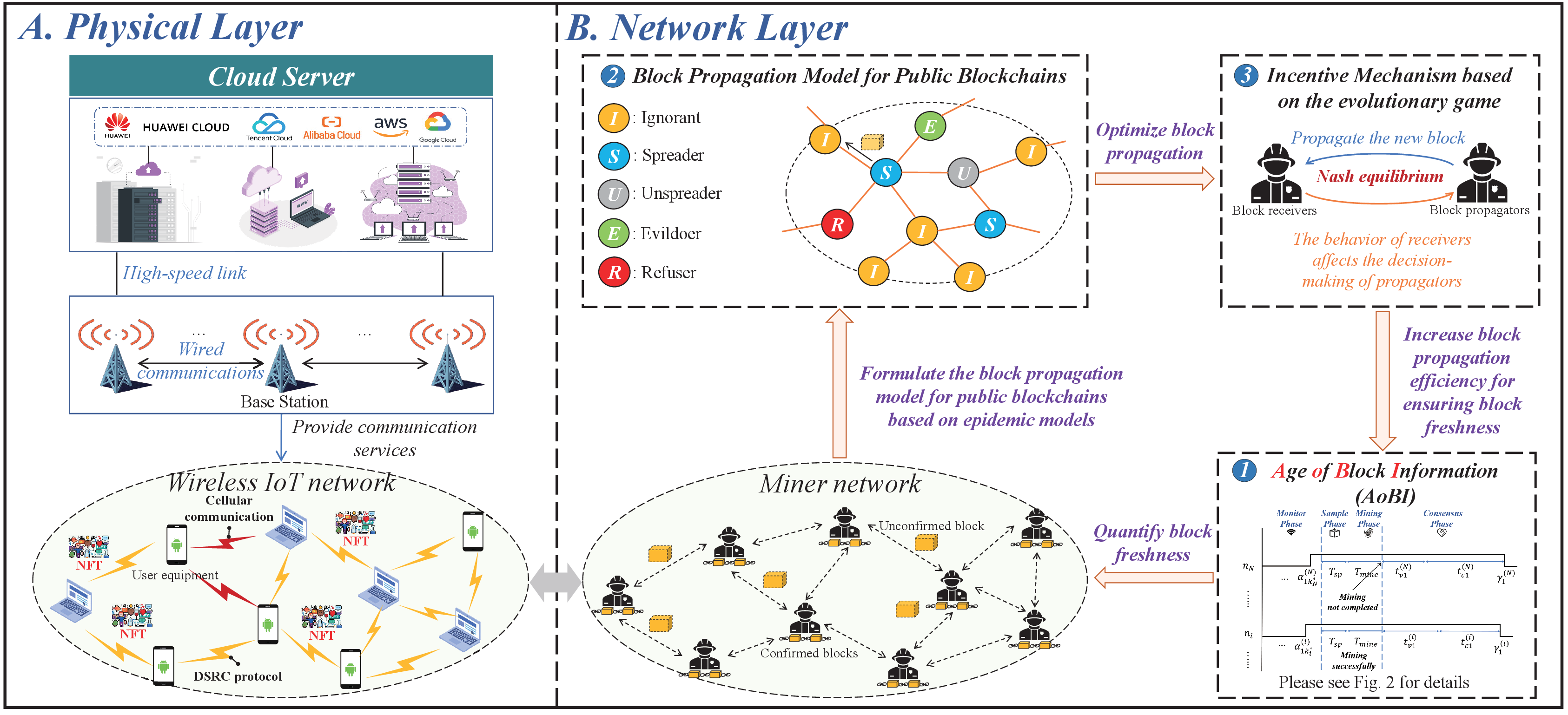}}
\captionsetup{font=footnotesize}
\caption{A freshness-aware block propagation optimization framework for Web 3.0.}     \label{system_model}     
\end{figure*}
In Web 3.0, we consider that the wireless blockchain network consists of a set $\mathcal{N} = \left\{n_1,\ldots, n_i,\ldots, n_N\right\}$ of $N$ IoT devices (i.e., miners), a set $\mathcal{B} = \left\{b_1,\ldots, b_j,\ldots, b_M\right\}$ of $M$ Base Stations (BSs), and sufficient cloud-based computing resources. Each IoT device makes updated data accessible to other devices by using the public blockchain technology that is a key driving force for enabling Web 3.0.


As shown in Fig. \ref{system_model}, miners can monitor physical data, such as NFT transactions and sensing data from the surrounding environment, and generate packets with newly observed data. Then, miners deliver the packets in the form of transactions with a timestamp to their own local \textit{mempools}. For miner $n_i$, it generates transactions independently and regularly at the time point $\alpha_{mk}^{(i)}$ under the Poisson distribution of rate $\lambda_{i}$\cite{rovira2019optimizing}, where $m, k \in \mathbb{Z^+}$ represent the $k$-th transaction generation in the $m$-th consensus. Therefore, the time interval between transactions with the same sequence in two consecutive consensuses follows the exponential distribution, i.e., $\Big(\Gamma_{m}^{(i)} = \alpha_{(m+1)k}^{(i)} - \alpha_{mk}^{(i)}\Big) \sim Exp(\lambda_{i})$ with the expected value $(1/\lambda_{i})$\cite{rovira2019optimizing}. 


To limit the waiting time of a transaction in the mempool, we define a packing period as $T_{p}$ and the maximum number of transactions in the block as $B_{max}$. During the packing period $T_{p}$, miners pick the freshest transactions with high transaction fees from their own local mempools and put the transactions into a block\cite{li2018transaction}. Then, miners aim to solve a cryptographic puzzle to obtain the bookkeeping right during the mining period $T_{mine}$\cite{nakamoto2008bitcoin}. Once a miner takes the lead in solving the cryptographic puzzle, its block $I_m$ (i.e., the $m$-th consensus block) will be forwarded immediately to the miner network for validation at the time point $\beta_m$. When a miner $n_l\in \mathcal{N}$ approves the new block $I_m$, the block will be available to its $k$ adjacent miners at the time point $\gamma_m^{(l)}$ after undergoing random validation time $t_{vm}^{(l)}$ and random communication time $t_{cm}^{(l)}$. Otherwise, the block will not be forwarded to avoid wasting network resources. We define the \textit{packing rate}\footnote{Note that the packing rate is the number of transactions packed into a block per second. Please refer to \url{https://cryptowallet.com/glossary/mempool/} for more details.} of miner $n_i$ as $\tau_i$ ($\mathrm{transaction/s}$)\cite{li2018transaction}. Due to insufficient energy and computing capacity of miners, the tasks of completing computation-based competitive consensus (e.g., Proof of Work \cite{nakamoto2008bitcoin, shi2022pooling}) and validating the new block require cloud-based computing resources\cite{jiao2019auction}. Moreover, the hash power has little difference among devices in the IoT network\cite{wang2023connectivity}. Thus, we consider that cloud-based computing resources allocated to each miner are $C\frac{\tau^2}{\sum_{i=1}^N\tau_i^2}$\footnote{Note that cloud servers allocate computing resources to miners based on the weight of $\tau^2$. On the one hand, the bigger $\tau$, the bigger $B_{size}$, making the computing resources allocated to miners more, which can avoid empty consensus blocks and restrain the low transaction rate of the system. On the other hand, the impact of $\tau$ on resource allocations can be increased.}\cite{qiu2019cloud}, where $C$ represents computing resources provided by a cloud computing server and $\tau$ represents the packing rate of the miner that obtains the bookkeeping right.

Although public blockchains can securely manage data and ensure data integrity for Web 3.0, they cannot guarantee the freshness of data packed into the block\cite{kim2022ensuring}. The use of outdated data for decision-making may cause incorrect outputs, further compromising the performance of the whole system\cite{kim2022ensuring}. Therefore, it is important to guarantee block freshness to ensure data freshness in Web 3.0. To measure block freshness, we propose a new metric called AoBI based on the concept of AoI. AoBI is defined as the time elapsed from the freshest transaction generation to the completion of block consensus, as shown in Fig. \ref{AoBI_component}.

\begin{figure}[t]
\centerline{\includegraphics[width=0.45\textwidth]{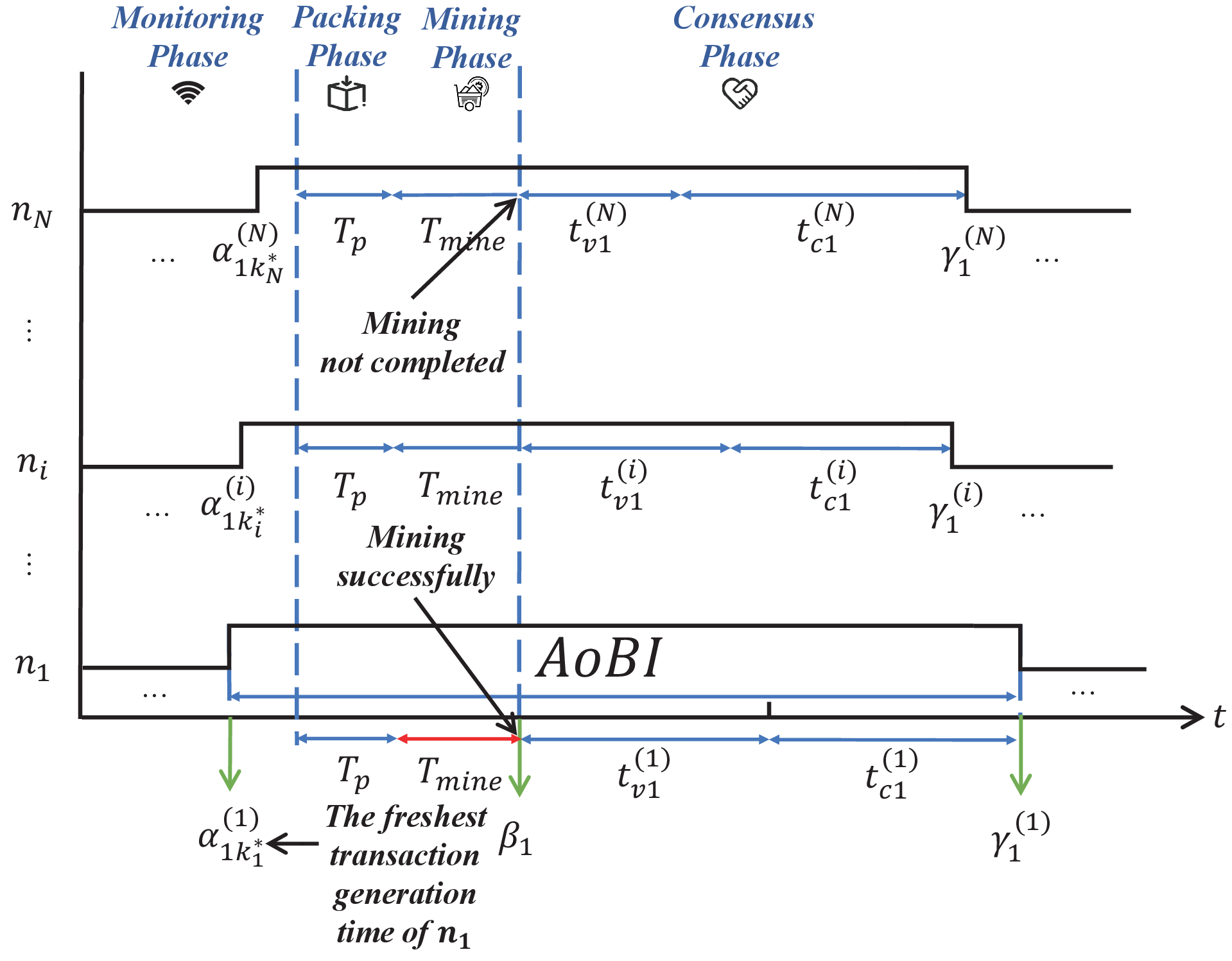}}
\captionsetup{font=footnotesize}
\caption{Age of Block Information for public blockchains.}
\label{AoBI_component}
\end{figure}

\section{Problem Formulation for Average AoBI Minimization}\label{AoBI_minmization}
In this section, we first propose a novel metric called AoBI to measure block freshness, which considers the block processing procedure, block validation, and block propagation. 

After mining, miners continue to monitor valuable data (e.g., NFT transactions) to ensure that valuable data can be updated in a timely manner. Therefore, we can obtain a constraint, which is given by
\begin{equation}
\begin{aligned}
    1/\lambda_{i} \geq T_{p} + T_{mine}.
\end{aligned}
\end{equation}

According to \cite{rovira2019optimizing}, the validation time $t_{vm}^{(i)}$ of miner $n_i$ is an exponentially distributed random variable $t_{vm}^{(i)} \sim Exp(\xi_i)$ with mean time $(1/\xi_i)$ depending on its computing capability. Thus, the average validation time of miner $n_i$ \cite{rovira2019optimizing} is\\
\begin{equation}\label{t_vm}
\begin{aligned}
    \mathbb{E}[t_{vm}^{(i)}] = 1/\xi_i &\geq \bigg[(R_v B_{size})\bigg/\bigg(C\frac{\tau^2}{\sum_{i=1}^N\tau_i^2}\bigg)\bigg]\\
    &\quad\: =\frac{R_vT_{p}\sum_{i=1}^N\tau_i^2}{C\tau},
\end{aligned}
\end{equation}
where $R_v$ represents the required number of instructions for a transaction to get validated by miners and $(B_{size}=\tau T_p)$ represents the number of transactions in the block. Note that the value range of $\tau$ is $\frac{1}{T_{p}} \leq \tau \leq \frac{B_{max}}{T_{p}}$ because of $1 \leq B_{size} \leq B_{max}$.

Based on the property of mean inequality, we know that
\begin{equation}
    \begin{aligned}
        \tau_1^2 + \tau_2^2 + \cdots + \tau_N^2 \geq \frac{(\tau_1 + \tau_2 + \cdots + \tau_N)^2}{N} = N\overline{\tau}^2,
    \end{aligned}
\end{equation}
if and only if $\tau_1 = \tau_2 = \cdots = \tau_N$, the equal sign holds. Considering $\tau_i$ following the uniform distribution in $\Big[0,\frac{B_{max}}{T_{p}}\Big]$, the average packing rate of miners is 
\begin{equation}
\begin{aligned}
    \overline{\tau} = \mathbb{E}[\tau_i] = \int_{0}^{\frac{B_{max}}{T_{p}}} \tau_i \frac{1}{\frac{B_{max}}{T_{p}}} \mathrm{d}\tau_i= \frac{B_{max}}{2T_{p}}.
\end{aligned}
\end{equation}
Thus, (\ref{t_vm}) can be rewritten as 
\begin{equation}
\begin{aligned}
    \mathbb{E}[t_{vm}^{(i)}] = 1/\xi_i \geq \frac{R_v N B_{max}^2}{4C\tau T_{p}}.
\end{aligned}
\end{equation}

The overall validation time of the network $T_{vm}$ is also an exponentially distributed random variable $T_{vm} \sim Exp(\Xi)$ with mean time $(1/\Xi)$ \cite{rovira2019optimizing}. We define $\overline{\omega}$ as the average density of adjacent miners that are willing to forward the block.

\begin{proposition}\label{P_1}
The average validation time of the miner network satisfies the following condition:
\begin{equation}
\mathbb{E}[T_{vm}] = 1/\Xi \geq \left\{
\begin{aligned}
    &\:  \Bigg\lceil \log_{\overline{\omega}k}\bigg(\frac{N(\overline{\omega}k-1)+k}{k}\bigg) \Bigg\rceil\frac{R_v N B_{max}^2}{4C\tau T_{p}}, \\
    &\qquad \qquad \qquad \quad\overline{\omega}k > 1,\\
    &\:  \bigg\lceil \frac{N}{k} \bigg\rceil \frac{R_v N B_{max}^2}{4C\tau T_{p}},\quad \overline{\omega}k = 1.
\end{aligned}
\right.
\end{equation}
\end{proposition}

\begin{proof}
    Please see Appendix (A).
\end{proof}


Similarly, it is shown that exponential distribution is a reasonable model for communication time, i.e., $t_{cm}^{(i)} \sim Exp(\eta_i)$, which includes the link establishment, actual transmission, and propagation through the network\cite{7021909}. Therefore, the overall communication time of the network $T_{cm}$ is also an exponentially distributed random variable $T_{vm} \sim Exp(H)$ with mean time $(1/H)$\cite{7021909}. On the one hand, the communications between any two IoT devices over a short distance or IoT devices and BSs adopt the Dedicated Short Range Communication (DSRC) protocol\cite{wang2022block}. On the other hand, the communications between any two IoT devices over a long distance need the relay of BSs. Since BSs use wired links to communicate with each other due to the fixed location, data propagation time between BSs can be negligible\cite{wang2022block}. If a BS $b_j$ covers a set $\mathcal{B}_j = \left\{n_{j1}, n_{j2},\ldots, n_{i},\ldots, n_{j{N_j}}\right\}\subset \mathcal{N}$ of $N_j$ miners, without loss generality, we consider that its total bandwidth $W_j$ is evenly assigned among the miners that are willing to forward the block\cite{rovira2019optimizing}. Thus, the average communication time of miner $n_i$ \cite{rovira2019optimizing} is given by
\begin{equation}
\begin{aligned}
    \mathbb{E}[t_{cm}^{(i)}] = 1/\eta_i &\geq \bigg[\left(\frac{P_{size}\tau T_{p}}{R_c}\right)\bigg/\left(\frac{W_j}{\overline{\omega} |\mathcal{B}_j|}\right)\bigg]\\
    &\quad \:= \frac{P_{size} \tau T_{p} \overline{\omega} |\mathcal{B}_j|}{R_c W_j},
\end{aligned}
\end{equation}
where $P_{size}(\textit{\rm{bit}})$ is the average size of a transaction and $R_c$ is the effective bit rate per unit bandwidth for the utilized networking technology.

According to \cite{rovira2019optimizing}, it is reasonable to consider that i) BSs are uniform, meaning that the total bandwidth of each BS is identical; ii) Miners and BSs are uniformly distributed, respectively; iii) The association between miners and BSs is uniform, meaning that for a network with $M$ BSs and $N$ miners, each BS covers $(N/M)$ miners, i.e., $|\mathcal{B}_j| = N/M$.

\begin{proposition}\label{P_2}
The average communication time of the miner network satisfies the following conditions:
\begin{equation}
\mathbb{E}[T_{cm}] = 1/H \geq \left\{
\begin{aligned}
    &\:  \Bigg\lceil \log_{\overline{\omega}k}\bigg(\frac{N(\overline{\omega}k-1)+k}{k}\bigg) \Bigg\rceil\frac{P_{size} \tau T_{p} \overline{\omega} N}{M R_c W}, \\
    &\qquad \qquad \qquad \quad\overline{\omega}k > 1,\\
    &\:  \bigg\lceil \frac{N}{k} \bigg\rceil \frac{P_{size} \tau T_{p} \overline{\omega} N}{M R_c W}, \quad \overline{\omega}k = 1.
\end{aligned}
\right.
\end{equation}
\end{proposition}
\begin{proof}
    Please see Appendix (A).
\end{proof}

In public blockchains, similar to \cite{ying2022aoti}, the AoBI consists of three parts: 1) The time from the freshest transaction generation to the mining end; 2) Block validation time; 3) Block propagation time. Therefore, the AoBI for the block $I_m$ and miner $n_i$ is given by
\begin{equation}\label{AoBI_I}
\begin{aligned}
    AoBI\left(I_m, n_i\right) &= \gamma_m^{(i)} - \max\left\{\alpha_{mk}^{(i)} \middle| \alpha_{mk}^{(i)} < \beta_m\right\}\\
    &= \left(\beta_m - \max\left\{\alpha_{mk}^{(i)} \middle| \alpha_{mk}^{(i)} < \beta_m\right\}\right)\\
    &\quad + t_{vm}^{(i)} + t_{cm}^{(i)}.
\end{aligned}
\end{equation}
The first term in (\ref{AoBI_I}) is the time passed since the freshest transaction generation $\alpha_{mk^*}^{(i)} = \max\left\{\alpha_{mk}^{(i)} \middle| \alpha_{mk}^{(i)} < \beta_m\right\}$ until the mining end $\beta_m$, where $\beta_m$ splits the interval between the two correspondingly freshest transaction generation $\left[\alpha_{mk^*}^{(i)}, \alpha_{(m+1)k^*}^{(i)}\right]$ into two intervals $\mu = \left[\alpha_{mk^*}^{(i)}, \beta_m\right]$ and $\nu = \left[\beta_m, \alpha_{(m+1)k^*}^{(i)}\right]$. Since $\left\{\beta_1, \beta_2, \ldots\right\}$ is independent of the transaction generation process $\left\{\alpha_{1j}^{(i)}, \alpha_{2j}^{(i)}, \ldots\right\}$, the cut point is uniformly distributed in $\left[\alpha_{mk^*}^{(i)}, \alpha_{(m+1)k^*}^{(i)}\right]$\cite{rovira2019optimizing}, and we can obtain $\mathbb{E}[\mu] = \mathbb{E}[\nu] = 1/(2\lambda_{i})$. Therefore, the average AoBI for miner $n_i$ is
\begin{equation}
\begin{aligned}
    \overline{AoBI}(n_i) &= \mathbb{E}\Big[\left(\beta_m - \max\left\{\alpha_{mk}^{(i)} \middle| \alpha_{mk}^{(i)} < \beta_m\right\}\right)\\ &\quad + t_{vm}^{(i)} + t_{cm}^{(i)}\Big] = \frac{1}{2\lambda_{i}} + \frac{1}{\xi_i} + \frac{1}{\eta_i}.
\end{aligned}
\end{equation}
Considering that the monitoring time of each miner is identical, the average AoBI for public blockchains is given by
\begin{equation}\label{AoBI}
\begin{aligned}
    \overline{AoBI} =\frac{1}{2\lambda}+ \mathbb{E}[T_{vm}] + \mathbb{E}[T_{cm}] = \frac{1}{2\lambda} + \frac{1}{\Xi} + \frac{1}{H}.
\end{aligned}
\end{equation}

To obtain the minimum value of average AoBI, we minimize (\ref{AoBI}) subject to the previously mentioned constraints as follows:

\begin{problem}\label{problem_1}
       When $\overline{\omega}k > 1$, the average AoBI minimization problem is
        \begin{equation}\label{A_pb_1}
        \begin{split}
        \overline{A}_{pb} &= \min_\tau\Big(\frac{1}{2\lambda} + \frac{1}{\Xi} + \frac{1}{H}\Big)\vspace{10ex}\\
        &\rm{s.t.}\left\{\begin{array}{lc}
        0 < \lambda \leq \frac{1}{T_{p}+T_{mine}},\vspace{1ex}\\
        0 < \Xi \leq \frac{4C\tau T_{p}}{\big\lceil \log_{\overline{\omega}k}\big(\frac{N(\overline{\omega}k-1)+k}{k}\big) \big\rceil R_v NB_{max}^2},\vspace{1ex}\\
        0 < H \leq \frac{MR_cW }{\big\lceil \log_{\overline{\omega}k}\big(\frac{N(\overline{\omega}k-1)+k}{k}\big) \big\rceil P_{size} \tau T_{p} \overline{\omega} N},
        \end{array}\right.
        \end{split}
        \end{equation}
        where $\frac{1}{T_{p}} \leq \tau \leq \frac{B_{max}}{T_{p}}$.
\end{problem}

\begin{problem}\label{problem_2}
         When $\overline{\omega}k = 1$, the average AoBI minimization problem is
        \begin{equation}\label{A_pb_2}
        \begin{split}
        \overline{A}_{pb} &= \min_\tau\Big(\frac{1}{2\lambda} + \frac{1}{\Xi} + \frac{1}{H}\Big)\vspace{10ex}\\
        &\rm{s.t.}\left\{\begin{array}{lc}
        0 < \lambda \leq \frac{1}{T_{p}+T_{mine}},\vspace{1ex}\\
        0 < \Xi \leq \frac{4C\tau T_{p}}{\big\lceil \frac{N}{k} \big\rceil R_v NB_{max}^2},\vspace{1ex}\\
        0 < H \leq \frac{MR_cW }{\big\lceil \frac{N}{k} \big\rceil P_{size} \tau T_{p} \overline{\omega} N},
        \end{array}\right.
        \end{split}
        \end{equation}
        where $\frac{1}{T_{p}} \leq \tau \leq \frac{B_{max}}{T_{p}}$.
\end{problem}

Note that $\lambda$, $\Xi$, and $H$ are independent of each other. When $\lambda$, $\Xi$, and $H$ take the maximum simultaneously, $\overline{A}_{pb}$ is minimum. Therefore, the minimum value of average AoBI for public blockchains is given by
\begin{equation}
    \overline{A}_{pb}(\tau) = \left\{
    \begin{split}
        & \Bigg\lceil \log_{\overline{\omega}k}\bigg(\frac{N(\overline{\omega}k-1)+k}{k}\bigg) \Bigg\rceil\bigg(\frac{P_{size}T_{p}\overline{\omega}N}{M R_c W }\tau +\\ &\frac{R_v N B_{max}^2}{4CT_{p}\tau}\bigg)
        +\frac{1}{2}(T_{p}+T_{mine}), \quad\overline{\omega}k > 1,\\
        & \bigg\lceil \frac{N}{k} \bigg\rceil\bigg(\frac{R_v N B_{max}^2}{4CT_{p}\tau} +  \frac{P_{size}T_{p}\overline{\omega}N}{M R_c W }\tau\bigg)\\
        & +\frac{1}{2}(T_{p}+T_{mine}), \quad\overline{\omega}k = 1,
    \end{split}
    \right.
\end{equation}
where $\frac{1}{T_{p}} \leq \tau \leq \frac{B_{max}}{T_{p}}$. To improve block freshness for ensuring the performance of blockchain-enabled Web 3.0, it is important to optimize block propagation and decrease the minimum value of average AoBI.

\section{Block Propagation Optimization}\label{optimize}
In this section, we focus on optimizing block propagation by improving block propagation efficiency, which can reduce consensus latency $(\mathbb{E}[T_{vm}] + \mathbb{E}[T_{cm}])$, thereby decreasing the minimum value of average AoBI.

\subsection{Block Propagation Model for Public Blockchains}
\subsubsection{Epidemic models}
Rumor dissemination models are generally built based on epidemic models, such as the Susceptible-Infected-Recovered (SIR) model\cite{zhao2012sihr,liu2016shir}. As one of the most classical epidemic models, the whole population in the SIR model is divided into three groups that are susceptible, infected, and recovered\cite{liu2016shir}. When contact with infected individuals, susceptible individuals check the disease and become infected states with a certain probability\cite{liu2016shir}. As time progresses, infected individuals no longer contract the disease anymore and become recovered states with a certain probability, indicating that infected individuals are cured or die\cite{liu2016shir}. {Since a new block is propagated in the form of rumor dissemination, the epidemic model is useful for modeling block propagation in public blockchains due to its ability to capture the behavior of information dissemination. Besides, the redundancy and fault-tolerant nature of the epidemic model contribute to alleviating the impact caused by the dynamics and complexity of the blockchain network.}

\subsubsection{Model formulation} Due to negligible propagation time between BSs, we ignore the transit process between BSs and regard the interaction between miners as a peer-to-peer interaction, which can better study the block propagation process. According to \cite{decker2013information}, when miners complete the block validation, they will forward the block to their adjacent miners if they approve the block. Otherwise, they will not forward the block. Finally, whether the block can be successfully added to blockchains depends on the approval results of all miners. {By analyzing the dynamic behavior of miners during block propagation\cite{ferdous2021survey}, miners can be divided into five groups based on epidemic models, namely \textit{Ignorants}, \textit{Spreaders}, \textit{Unspreaders}, \textit{Refusers}, and \textit{Evildoers}. Specifically, ignorants are initial miners that have not received the new block. Spreaders are miners that approve the block and forward it to their $k$ adjacent miners. Unspreaders are miners that do not approve the block and discard it, which can avoid wasting network resources. Refusers are immune miners that will not receive and forward the block\footnote{Note that spreaders will send \textit{inv} messages to their adjacent miners before forwarding the block\cite{decker2013information}. When refusers receive inv messages, they will not issue \textit{getdata} messages, which indicates that they have completed the block validation and no longer receive the block\cite{vu2019efficient}. Note that unspreaders are the transition status from ignorants to refusers.}. Evildoers are malicious miners that destroy the interests of most miners in the network, e.g., by not forwarding the block deliberately 
\cite{sallal2022security}.} Based on \cite{zhao2012sihr}, the state transition diagram of the block propagation model is shown in Fig. \ref{s_trans}, and the specific conversion rules are described as follows:
\begin{figure}[h]
\centerline{\includegraphics[width=0.45\textwidth]{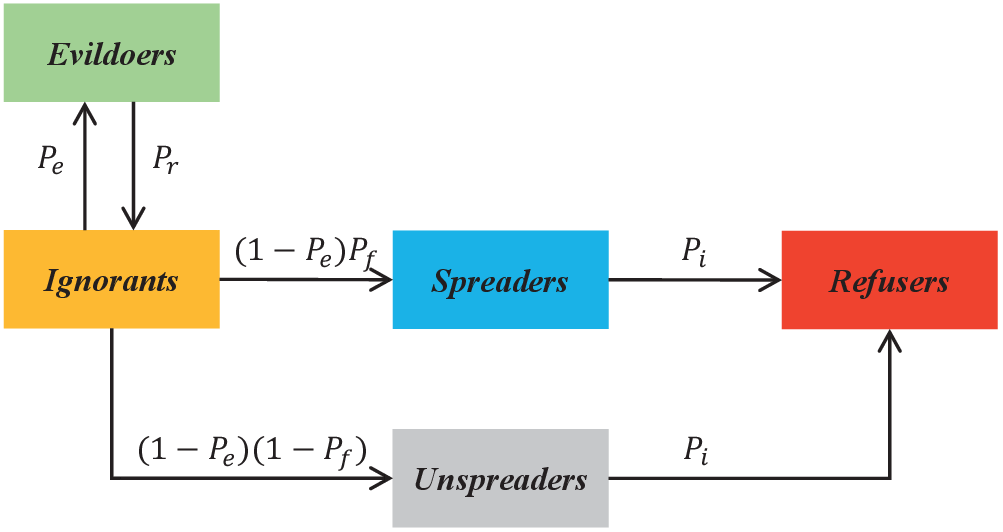}}
\captionsetup{justification=raggedright,singlelinecheck=false, font=footnotesize}
\caption{The state transition diagram of the block propagation model for public blockchains.}
\label{s_trans}
\end{figure}
\begin{itemize}
    \item When ignorants receive a block forwarded from spreaders, the ignorants will validate the block and be converted to spreaders with a probability $P_{f}$. Otherwise, the ignorants will be converted to unspreaders with a probability $(1-P_{f})$, where $P_{f} \in [0, 1]$ is the forwarding probability.
    
    \item When each round of interaction starts, ignorants will be converted to evildoers with a probability $P_{e}$. When each round of interaction ends, the evildoers will be converted to ignorants with a probability $P_{r}$, where $P_{e}\in [0,1]$ is the evil probability and $P_{r}\in [0,1]$ is the recovery probability.
    
    \item When spreaders receive the block again or after each interaction, the spreaders will be converted to refusers with a probability $P_{i}$. Similarly, unspreaders will be converted to refusers with a probability $P_i$ after each interaction, where $P_{i}\in(0,1]$ is the immunity probability.
\end{itemize}

We denote $i(t)$, $s(t)$, $u(t)$, $r(t)$, and $e(t)$ as the proportion of ignorants, spreaders, unspreaders, refusers, and evildoers at time $t$, respectively. They satisfy the following normalization condition\cite{zhao2012sihr}: 
\begin{equation}
i(t)+s(t)+u(t)+r(t)+e(t)=1.
\end{equation}
According to Fig. \ref{s_trans}, the variation of $i(t)$ at $\Delta t$ can be obtained as follows:
\begin{dmath}
N\big[i(t+\Delta t)-i(t)\big]=P_{r} Ne(t)\Delta t-P_{e} Ni(t)\Delta t-kP_{f}(1-P_{e})Ns(t)i(t)\Delta t-k(1-P_{f})(1-P_{e})Ns(t)i(t)\Delta t,
\end{dmath}
that is
\begin{dmath}\label{i(t)}
\frac{i(t+\Delta t)-i(t)}{\Delta t}=P_{r} e(t)-P_{e} i(t)-kP_{f}(1-P_{e})s(t)i(t)-k(1-P_f)(1-P_{e})s(t)i(t).
\end{dmath}
Taking the limit of $\Delta t \rightarrow 0$ on both sides  simultaneously, we can obtain the mean-field equation of ignorants as 
\begin{dmath}
    \frac{\mathrm{d}i(t)}{\mathrm{d}t}=P_{r} e(t)-P_{e} i(t)-kP_{f}(1-P_{e}) s(t)i(t)\\
    -k(1-P_f)(1-P_{e})s(t)i(t) = P_{r} e(t)-P_{e} i(t)-k(1-P_{e}) s(t)i(t).
\end{dmath}
Using the same method, the mean-field equations of the proposed block propagation can be described as follows:
\begin{align}
    \frac{\mathrm{d}i(t)}{\mathrm{d}t}&=P_{r} e(t)-P_{e} i(t)-k(1-P_{e}) s(t)i(t),\\
    \frac{\mathrm{d}s(t)}{\mathrm{d}t}&=kP_{f}(1-P_{e}) s(t)i(t)-P_{i}\big(1+ks(t)\big)s(t),\\
    \frac{\mathrm{d}u(t)}{\mathrm{d}t}&=k(1-P_f)(1-P_{e})s(t)i(t)-P_{i}u(t),\\
    \frac{\mathrm{d}r(t)}{\mathrm{d}t}&=P_{i}\big(1+ks(t)\big)s(t)+P_iu(t),\\
    \frac{\mathrm{d}e(t)}{\mathrm{d}t}&=P_{e}i(t)-P_{r} e(t).
\end{align}

\begin{figure*}[t]
	\centering
        \captionsetup{font=footnotesize}
	\subfloat[ The evil probability $P_e$ is fixed.]
	{\includegraphics[width=0.3\textwidth]{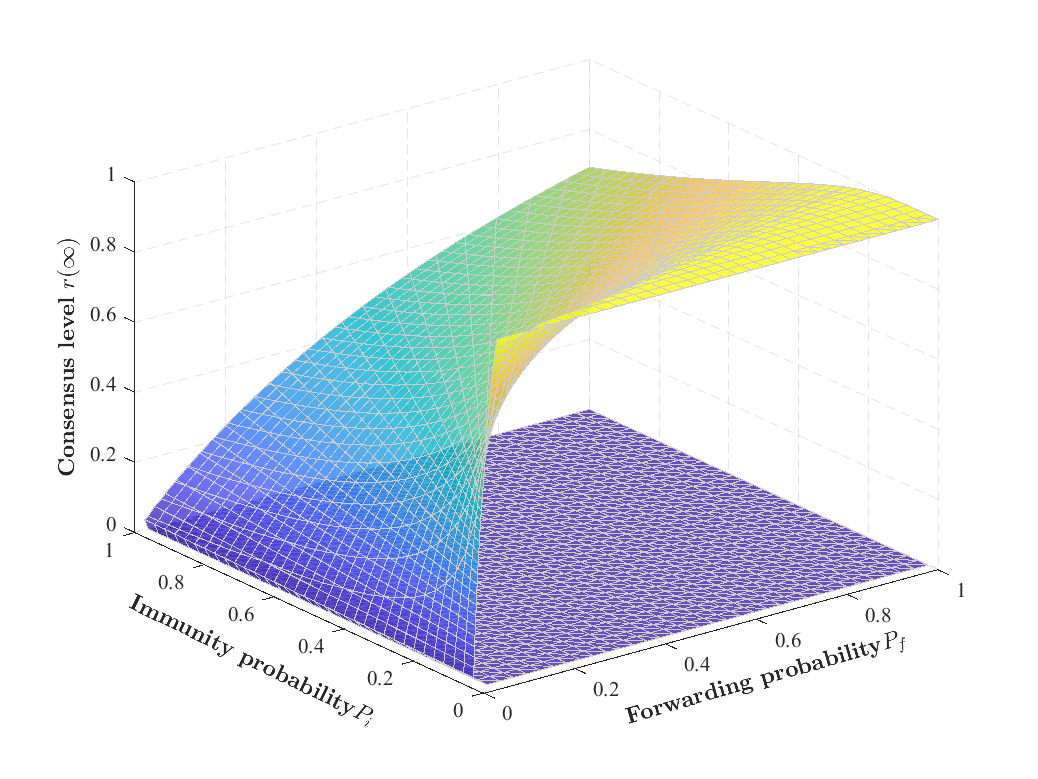}\label{p_e_fixed}}
        \captionsetup{font=footnotesize}
	\subfloat[ The immunity probability $P_i$ is fixed.]
	{\includegraphics[width=0.3\textwidth]{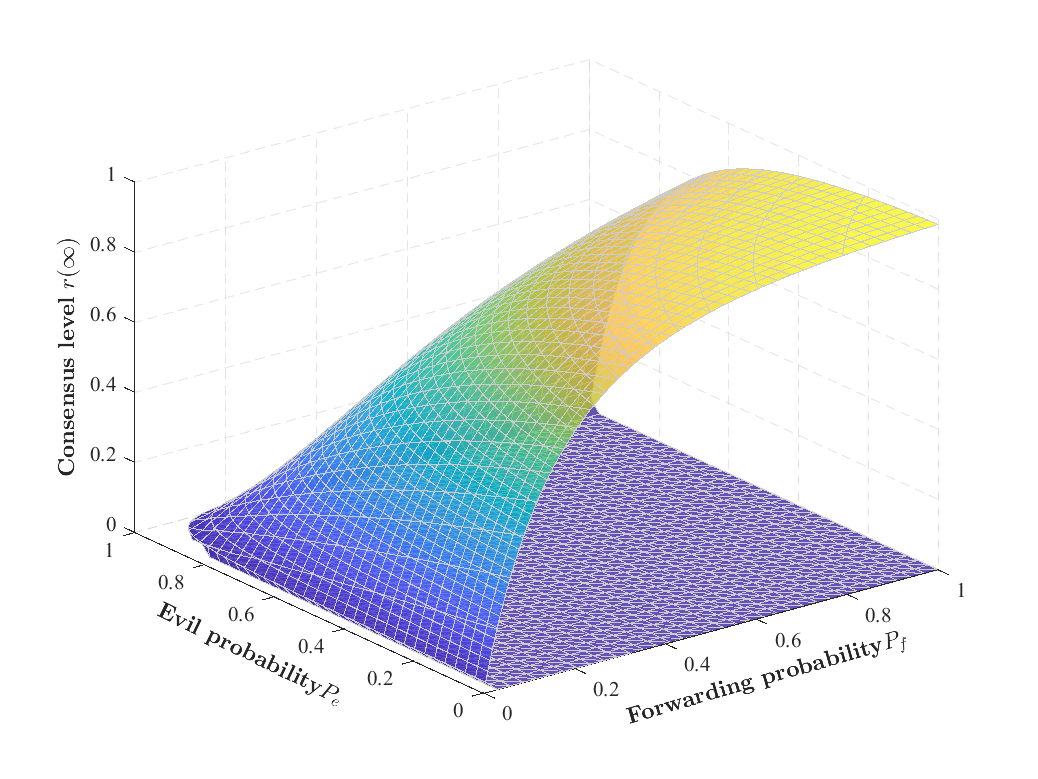}\label{p_i_fixed}}
        \captionsetup{font=footnotesize}
	\subfloat[ The forwarding probability $P_f$ is fixed.]
	{\includegraphics[width=0.3\textwidth]{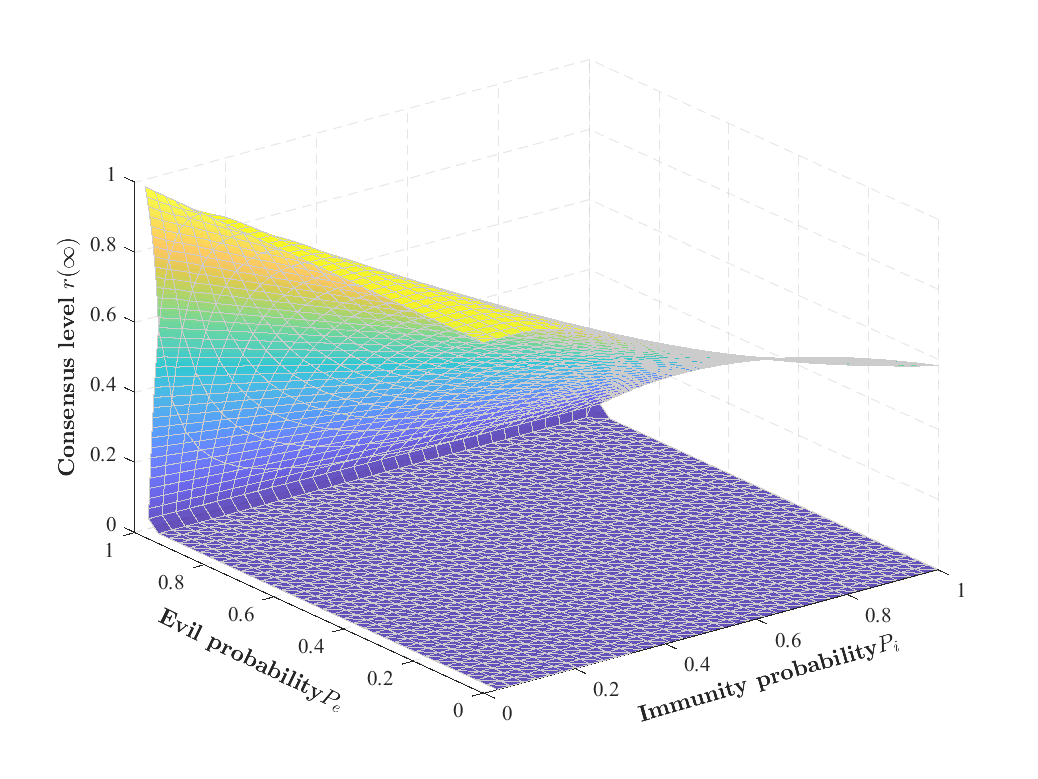}\label{p_f_fixed}}	
        \captionsetup{font=footnotesize}
	\caption{Consensus level $r(\infty)$ variations under the other two probabilities when one probability is fixed.}\label{r_inifty}
\end{figure*}

The \textit{Runge–Kutta} algorithm can be used to solve the above differential equations and analyze the effects on block propagation by important factors\cite{zhao2012sihr}. The computational complexity of the algorithm is $\mathcal{O}(h^5)$, where $h$ is a time step. The initial value of each proportion is
\begin{gather}
   i(t)=\frac{N-1}{N}, s(t)=\frac{1}{N}, u(t)=0, r (t)=0, e(t)=0.
\end{gather}
\subsubsection{Steady-state analysis} In the whole process of block propagation, spreaders facilitate block propagation. The number of spreaders first increases, then decreases, and reaches zero when the block consensus is done. At that time, the miner network reaches an equilibrium state, and the steady state of the network consists of ignorants, evildoers, and refusers. We analyze the final size of block consensus $r(\infty)$, where $r(\infty) = final\left\{r(t)\right\} = \lim_{t\rightarrow\infty}r(t)$\cite{zhao2012sihr,liu2016shir}. $r(\infty)$ can be used to measure the level of consensus. Taking $r(\infty) = 0.8$ as an example, it means that $80\%$ of miners have validated the block in the end. $r(\infty)$ is the solution of the transcendental equation\cite{liu2016shir}. According to the property of the transcendental equation, $r(\infty) = 1-e^{-\sigma r(\infty)}$, where $\sigma = \frac{(1-P_e)P_f}{P_i}+1$.

\begin{theorem}\label{theorem_1}
If $\sigma > 1$, namely $P_e \neq 1$ and $P_f \neq 0$, the consensus level $r(\infty) = 1-e^{-\sigma r(\infty)}$ has two solutions: zero and a nontrivial solution $R$, where $R \in (0,1)$.
\end{theorem}
\begin{proof}
It is obvious that $r(\infty) = 0$ is a solution of $r(\infty) = 1-e^{-\sigma r(\infty)}$. We construct a function $f(x) = x + e^{-\sigma x} - 1$. By taking the second order derivative of $f(x)$, we can get $f^{\prime \prime}(x) = \sigma^2e^{-\sigma x}>0$. Thus, $f(x)$ is a convex function. Since $f^\prime(0) = 1 - \sigma e^{-\sigma x}|_{x = 0} = 1 - \sigma < 0$ and $f(1) = e^{-\sigma} > 0$, the function $f(x)$ must exist a nontrivial solution $R$, where $0 < R < 1$. Therefore, the theorem is proved.
\end{proof}

\begin{theorem}
If $\sigma > 1$, namely $P_e \neq 1$ and $P_f \neq 0$, there is no consensus threshold in the block propagation model.
\end{theorem}
\begin{proof}
Based on \textbf{Theorem \ref{theorem_1}}, $r(\infty) = 0$ is a solution of $r(\infty) = 1-e^{-\sigma r(\infty)}$. Building a function $g(x) = 1 - e^{-\sigma x}$ and taking the second order derivative of $g(x)$, we can get $g^{\prime \prime}(x) = -\sigma^2e^{-\sigma x} < 0$. Thus, $g(x)$ is a concave function. Since $g^\prime(0) = \sigma e^{-\sigma x}|_{x = 0} = \sigma > 1$ and $g^\prime(0)$ can be any value for all values of probabilities $P_e$, $P_f$, and $P_i$, $r(\infty)$ has no maximum. Therefore, the theorem is proved.  
\end{proof}

\begin{theorem}
Given fixed $P_e$ and $P_i$, $r(\infty)$ increases as $P_f$ increases. Given fixed $P_e$ and $P_f$, $r(\infty)$ decreases as $P_i$ increases. Similarly, given fixed $P_i$ and $P_f$, $r(\infty)$ decreases as $P_e$ increases.
\end{theorem}
\begin{proof}
As shown above, the consensus level of the miner network is given by
\begin{equation}
    \begin{aligned}
        r(\infty) = 1 - e^{-\Big(\frac{(1-P_e)P_f}{P_i}+1\Big)r(\infty)}.
    \end{aligned}
\end{equation}

When $P_e$ and $P_i$ are fixed, we take the derivative of $r(\infty)$ with respect to $P_f$ and get

\begin{equation}\label{r_f}
    \begin{aligned}
        r^\prime(\infty) = \frac{\frac{1-P_e}{P_i}e^{-\Big(\frac{(1-P_e)P_f}{P_i}+1\Big)r(\infty)}r(\infty)}{1-\Big(\frac{(1-P_e)P_f}{P_i} + 1\Big)e^{-\Big(\frac{(1-P_e)P_f}{P_i}+1\Big)r(\infty)}}.
    \end{aligned}
\end{equation}
Note that the numerator of (\ref{r_f}) is greater than $0$. We construct a function $h(x) = 1 - \epsilon e^{-\epsilon x}$. Taking the derivative of $h(x)$, we can get $h^\prime(x) = \epsilon^2e^{-\epsilon x} > 0$. Thus, $h(x)$ is a monotonically increasing function. Since $h(0) = 1 > 0$, we have $h(x) > 0, x \in (0,1)$, and $r^\prime(\infty) > 0$. Therefore, the first part of this theorem is proved.

When $P_e$ and $P_f$ are fixed, we take the derivative of $r(\infty)$ with respect to $P_i$ and get
\begin{equation}\label{r_i}
    \begin{aligned}
        r^\prime(\infty) = \frac{-\frac{(1-P_e)P_f}{P_i^2}e^{-\Big(\frac{(1-P_e)P_f}{P_i}+1\Big)r(\infty)}r(\infty)}{1-\Big(\frac{(1-P_e)P_f}{P_i} + 1\Big)e^{-\Big(\frac{(1-P_e)P_f}{P_i}+1\Big)r(\infty)}}.
    \end{aligned}
\end{equation}
From the first part of this theorem, we know that the denominator of (\ref{r_i}) is greater than $0$. Since the numerator of (\ref{r_i}) is less than $0$, $r^\prime(\infty) < 0$. Therefore, the second part of this theorem is proved.

When $P_i$ and $P_f$ are fixed, we take the derivative of $r(\infty)$ with respect to $P_e$ and get
\begin{equation}\label{r_e}
    \begin{aligned}
        r^\prime(\infty) = \frac{-\frac{P_f}{P_i}e^{-\Big(\frac{(1-P_e)P_f}{P_i}+1\Big)r(\infty)}r(\infty)}{1-\Big(\frac{(1-P_e)P_f}{P_i} + 1\Big)e^{-\Big(\frac{(1-P_e)P_f}{P_i}+1\Big)r(\infty)}}.
    \end{aligned}
\end{equation}
Since the numerator of (\ref{r_e}) is less than $0$ and the denominator of (\ref{r_e}) is greater than $0$, we have $r^\prime(\infty) < 0$. Therefore, the third part of this theorem is proved.
\end{proof}

The above analysis indicates that the consensus level of the miner network $r(\infty)$ is a function of the forwarding probability $P_f$, the immunity probability $P_i$, and the evil probability $P_e$. Figure \ref{r_inifty} shows variations of the consensus level $r(\infty)$ as changes in the other two probabilities when one probability is fixed. From Fig. \ref{p_e_fixed}, we can observe that when the evil probability $P_e$ is fixed, the consensus level $r(\infty)$ has great changes, and almost any values from $0$ to $1$ can be taken. Given a fixed $P_e$, $r(\infty)$ increases as $P_i$ decreases, but it looks like that $r(\infty)$ first decreases and then increases as $P_i$ decreases. The reason for having this wrong vision is that the surface of Fig. \ref{p_e_fixed} is distorted. Considering that propagating a block causes the energy cost of miners, incentives can be used to encourage miners to propagate the block actively.

\begin{table*}[t]\small\label{Matrix}
    \renewcommand\arraystretch{1.1} 
    \centering
    \captionsetup{font=footnotesize}
    \caption{The Evolutionary Game Payoff Matrix for Block Propagators and Block Receivers in Public Blockchains.}
    \begin{tabular}{c|cc}
    \hline
    \diagbox{Propagator strategy}{Receiver strategy} & Forwarding $y(t)$ & Not forwarding $(1-y(t))$  \\
    \hline 
    Forwarding $x(t)$ & $(\Delta I+\Delta U+\Delta P,\Delta I+\Delta U+\Delta P)$ & $(\Delta U + \Delta P -\varepsilon R,\Delta P)$ \\
    Not forwarding $(1-x(t))$ & $-$ & $(\Delta P,0)$\\
    \hline
    \end{tabular}\label{Matrix}
\end{table*}

\subsection{Incentive Mechanism based on Evolutionary Game for Block Propagation}

{Since the wireless IoT network has the characteristics of large coverage and heterogeneity\cite{wang2023connectivity, 9454291}, miners find it difficult to choose the best strategy for block propagation in the complex environment to maximize their benefits. Hence, miners often make near-optimal decisions based on local information they have. Evolutionary games are time-varying decisions that consider dynamic scenarios with time-varying parameters, close to the real situation\cite{weibull1997evolutionary}. Based on the evolutionary game, miners can only consider limited information in the block propagation process, and the behaviors of miners will evolve to the final stable state in the process of continuous trials and errors\cite{4607241}.} Based on the Evolutionary Stable Strategy (ESS)\cite{weibull1997evolutionary}, we can formulate reasonable optimization strategies to reduce the number of meaningless block transmissions and improve block propagation efficiency, thereby decreasing the minimum value of average AoBI.

It is worth mentioning that block propagation efficiency has a negative correlation with $\mathbb{E}[T_{cm}]$. To be specific, block propagation efficiency is the effective propagation rate of the miner network. When communication time between miners is fixed, the bigger $\mathbb{E}[T_{cm}]$, the larger the number of meaningless block transmissions, making block propagation efficiency lower. Therefore, the lower the minimum value of average AoBI, the higher the block propagation efficiency. 

To improve block propagation efficiency, we propose an Incentive Mechanism based on the evolutionary game for Block Propagation (called BPIM)\cite{wenoptimal}. We study changes in the forwarding probability under the role of the incentive mechanism and find the optimal strategy combination for miners, which can optimize the block propagation mechanism, thereby better achieving the whole network consensus. Considering that all miners are bounded rational, block propagation decisions simultaneously move in the game where one party takes an action without knowing the strategy the other party is taking, namely when deciding on actions, it is inferred that other parties will also act rationally.

\subsubsection{Payoff matrix of the evolutionary game}
Miners can be essentially divided into two groups that are block propagators (i.e., spreaders) and block receivers (i.e., ignorants, spreaders, and evildoers)\cite{ferdous2021survey}. Note that their characteristics are consistent, which indicates that a miner can be either a block propagator or a block receiver.
 

For block validation, the reward and cost are defined as $P$ and $Q$, respectively. For block propagation, the reward and cost are defined as $I$ and $M$, respectively. Moreover, we define $\Delta I \in \mathbb{R}^+$ as an extra block propagation reward for both block propagators and receivers forwarding the block. Since block validation is executed before block propagation, we consider $P>Q$. To facilitate research, we define $\Delta P = (P - Q) \in \mathbb{R}^+$ as the basic validation reward and $\Delta U = (I - M) \in \mathbb{R}$ as the basic propagation reward. Besides, we define $R\in \mathbb{R}^+$ as the punishment risk for the spreaders that forward the block to evildoers and $\varepsilon > 0$ as the unit cost for the punishment risk. Thus, $\varepsilon R$ is the cost that spreaders forward the block to evildoers. Then, we establish an evolutionary game payoff matrix for block propagators and receivers in public blockchains, where the action strategies of miners are whether to forward the new block, as shown in Table \ref{Matrix}. We analyze the cases of the evolutionary game as follows:

\begin{itemize}
    \item When both block propagators and receivers forward the block, both of them will not only receive the block propagation reward $I$ but also the extra block propagation reward $\Delta I$, in which the revenue functions of block propagators are $(\Delta I+\Delta U + \Delta P)$ and the revenue functions of block receivers are $(\Delta I+\Delta U + \Delta P)$.
    \item When block propagators forward the block, but block receivers do not forward the block, in this case, the block receivers may be evildoers. Thus, the block propagators need to undertake the cost $(\varepsilon R)$, in which the revenue functions of block propagators are $(\Delta U + \Delta P - \varepsilon R)$ and the revenue functions of block receivers are $\Delta P$.
    \item When block propagators do not forward the block, but block receivers forward the block, the logic does not hold. Therefore, this case does not exist.
    \item When both block propagators and receivers do not forward the block, the block propagators only obtain the basic validation reward. Thus, the revenue functions of block propagators are $\Delta P$ and the revenue functions of block receivers are $0$.
\end{itemize}

We define that the forwarding probability of block propagators and block receivers are $x(t)$ and $y(t)$, respectively, where $x(t)$ and $y(t)$ are consistent with the forwarding probability $P_{f}$ of the block propagation model for public blockchains with respect to time $t$. Based on the payoff matrix for block propagators and block receivers, we can obtain the revenue functions of block propagators and receivers.

\subsubsection{Revenue functions of block participants}
For block propagators, the expected revenue of forwarding the block $G_{1Y}$ \cite{4607241} is given by
\begin{dmath}\label{G_1Y}
     G_{1Y}=y(t)(\Delta I+\Delta U + \Delta P)+(1-y(t))(\Delta U + \Delta P - \varepsilon R) =y(t)\Delta I + \Delta U + \Delta P - (1-y(t))\varepsilon R.
\end{dmath}
Similarly, the expected revenue of not forwarding the block $G_{1N}$ is given by 
\begin{dmath}
    G_{1N}=(1-y(t))\Delta P.
\end{dmath}
Thus, the group average revenue of block propagators $G_1$ \cite{4607241} can be expressed as
\begin{dmath}\label{G_1}
     G_1 = x(t)G_{1Y} + (1-x(t))G_{1N}=x(t)y(t)(\Delta I + \Delta P + \varepsilon R) + x(t)(\Delta U - \varepsilon R) + (1-y(t))\Delta P.
\end{dmath}
For block receivers, the expected revenue of forwarding the block $G_{2Y}$ is given by
\begin{gather}\label{G_2Y}
    G_{2Y}=x(t)(\Delta I + \Delta U + \Delta P),
\end{gather}
and the expected revenue of not forwarding the block $G_{2N}$ is given by
\begin{gather}
    G_{2N}=x(t)\Delta P.
\end{gather}
Thus, the group average revenue of block receivers $G_2$ can be expressed as
\begin{gather}\label{G_2}
\begin{split}
    G_{2} &= y(t)G_{2Y}+(1-y(t))G_{2N}\\&=x(t)y(t)(\Delta I + \Delta U) + x(t)\Delta P.
\end{split}
\end{gather}

\subsubsection{Replicator dynamics for the forwarding probability}
Based on (\ref{G_1Y}) and (\ref{G_1}), the replicator dynamic for the forwarding probability of block propagators can be expressed by the following system of Ordinary Differential Equations (ODEs)\cite{liu2018evolutionary}:
\begin{gather}\label{H(x)}
\begin{split}
    H(x(t))&=\dot{x}(t)=x(t)(G_{1Y}-G_{1})\\
    &=x(t)(1-x(t))[y(t)(\Delta I+\Delta P +\varepsilon R) + \Delta U -\varepsilon R].
\end{split}
\end{gather}
If $y(t)\equiv\frac{-\Delta U + \varepsilon R}{\Delta I + \Delta P + \varepsilon R}$, then $H(x(t))\equiv 0$, indicating that the forward probability of block propagators does not change. If $y(t)\ne \frac{-\Delta U + \varepsilon R}{\Delta I + \Delta P + \varepsilon R}$, let $H(x(t))\equiv0$, then $x(t)\equiv0$ and $x(t)\equiv1$ are possible equilibrium values of $x(t)$, and the derivative of $H(x(t))$ with respect to $x(t)$ is given by
    \begin{dmath}
        \frac{\mathrm{d}H(x(t))}{\mathrm{d}x(t)}=(1-2x(t))[y(t)(\Delta I+\Delta P + \varepsilon R)+ \Delta U -\varepsilon R].
    \end{dmath}
    
To satisfy the condition of ESS, i.e.,  $\frac{\mathrm{d}H(x(t))}{\mathrm{d}x(t)}<0$, we compare the size of $y(t)$ and $\frac{-\Delta U + \varepsilon R}{\Delta I + \Delta P + \varepsilon R}$, and we can obtain the following theorems:
\begin{theorem}
When $y(t) > \frac{-\Delta U + \varepsilon R}{\Delta I + \Delta P + \varepsilon R}$, if $\frac{\mathrm{d}H(x(t))}{\mathrm{d}x(t)}<0$, then $x(t)\equiv1$ is the only possible equilibrium value of the forwarding probability of block propagators. When $(\Delta I + \Delta P + \Delta U)\gg 0$, then $\frac{-\Delta U + \varepsilon R}{\Delta I + \Delta P + \varepsilon R}\approx 0$. Thus, $y(t)>\frac{-\Delta U + \varepsilon R}{\Delta I + \Delta P + \varepsilon R}$ is always true. Similarly, when $\Delta U > \varepsilon R$, then $\frac{-\Delta U + \varepsilon R}{\Delta I + \Delta P + \varepsilon R} < 0$. Therefore, $y(t)>\frac{-\Delta U + \varepsilon R}{\Delta I + \Delta P + \varepsilon R}$ is always true.
\end{theorem}

\begin{theorem}
When $y(t)<\frac{-\Delta U + \varepsilon R}{\Delta I + \Delta P + \varepsilon R}$, if $\frac{\mathrm{d}H(x(t))}{\mathrm{d}x(t)}<0$, then $x(t)\equiv0$ is the only possible equilibrium value of the forwarding probability of block propagators. When $(\Delta I + \Delta P + \Delta U) < 0$, we have $\frac{-\Delta U + \varepsilon R}{\Delta I + \Delta P + \varepsilon R} > 1$. Therefore, $y(t) < \frac{-\Delta U + \varepsilon R}{\Delta I + \Delta P + \varepsilon R}$ is always true.
\end{theorem}

Similarly, based on (\ref{G_2Y}) and (\ref{G_2}), the replicator dynamic for the forwarding probability of block receivers can be expressed as\cite{liu2018evolutionary}
\begin{gather}\label{H(y)}
 \begin{split}
    H(y(t))&=\dot{y}(t)=y(t)(G_{2Y}-G_{2})\\
    &=x(t)y(t)(1-y(t))(\Delta I + \Delta U).
 \end{split}
\end{gather}
If $(\Delta I + \Delta U) = 0$ or $x(t)\equiv0$, then $H(y(t))\equiv 0$, indicating that the forward probability of block receivers does not change.
If $(\Delta I + \Delta U) \neq 0$ and $x(t) \ne 0$, let $H(y(t))\equiv0$, then $y(t)\equiv0$ and $y(t)\equiv1$ are possible equilibrium values of $y(t)$, and the derivation of $H(y(t))$ with respect to $y(t)$ is
    \begin{gather}
      \begin{split}
        \frac{\mathrm{d}H(y(t))}{\mathrm{d}y(t)}=x(t)(1-2y(t))(\Delta I + \Delta U).
      \end{split}
    \end{gather}
Since $x(t)>0$, the following theorems can be obtained as

\begin{figure*}[t]
	\hrulefill
        \begin{equation}\label{jm}
        \bm{J}=
        \begin{pmatrix}
	    \frac{\partial H(x(t))}{\partial x(t)} & \frac{\partial H(x(t))}{\partial y(t)}\\
            \frac{\partial H(y(t))}{\partial x(t)} & \frac{\partial H(y(t))}{\partial y(t)}
	\end{pmatrix}
        =
        \begin{pmatrix}
            (1-2x(t))[y(t)(\Delta I + \Delta P + \varepsilon R)+\Delta U - \varepsilon R] & x(t)(1-x(t))(\Delta I + \Delta P + \varepsilon R)\\
            y(t)(1-y(t))(\Delta I + \Delta U) & x(t)(\Delta I + \Delta U)(1-2y(t))
        \end{pmatrix}.
        \end{equation}
\end{figure*}

\begin{theorem}
When $(\Delta I + \Delta U) < 0$, if $\frac{\mathrm{d}H(y(t))}{\mathrm{d}y(t)}$, then $y(t)\equiv0$ is the only possible equilibrium value of the forwarding probability of block receivers.
\end{theorem}
\begin{theorem}
When $(\Delta I + \Delta U) > 0$, if $\frac{\mathrm{d}H(y(t))}{\mathrm{d}y(t)}$, then $y(t)\equiv1$ is the only possible equilibrium value of the forwarding probability of block receivers.
\end{theorem}

Based on the above analyses of game solutions, we can obtain three possible equilibrium points in the evolutionary game, i.e., $(0,0)$, $(1,0)$, and $(1,1)$. Then, we analyze the evolutionary stability of the equilibrium points.

\subsubsection{Game equilibrium analysis}
In the evolutionary game, the Jacobian matrix of the replicator dynamics can be used to validate the evolutionary stability of equilibrium points\cite{liu2018evolutionary}. Specifically, the equilibrium point is stable if $\mathrm{det}(\bm{J})=\frac{\partial H(x(t))}{\partial x(t)}\frac{\partial H(y(t))}{\partial y(t)}-\frac{\partial H(x(t))}{\partial y(t)}\frac{\partial H(y(t))}{\partial x(t)}>0$ and $\mathrm{tr}(\bm{J})=\frac{\partial H(x(t))}{\partial x(t)}+\frac{\partial H(y(t))}{\partial y(t)}<0$. Based on (\ref{jm}), we analyze the stability of three equilibrium points as follows:
\begin{enumerate} [a)]
    \item For the equilibrium point $(0,0)$, we have $\mathrm{det}(\bm{J})|_{(0,0)}=0$ and $\mathrm{tr}(\bm{J})|_{(0,0)}=\Delta U - \varepsilon R$. Therefore, the equilibrium point $(0,0)$ is a saddle point rather than an ESS point\cite{weibull1997evolutionary}. 
    \item For the equilibrium point $(1,0)$, we have $\mathrm{det}(\bm{J})|_{(1,0)}=0$ and $\mathrm{tr}(\bm{J})|_{(1,0)}=\Delta I + \Delta P + \Delta U$. Therefore, the equilibrium point $(1,0)$ is also a saddle point.
    \item For the equilibrium point $(1,1)$, we have $\mathrm{det}(\bm{J})|_{(1,1)}=(\Delta I + \Delta U)(\Delta I + \Delta P + \Delta U)$ and $\mathrm{tr}(\bm{J})|_{(1,1)}=-2(\Delta I + \Delta U)-\Delta P$. When $(\Delta I + \Delta U) > 0$, we have $\mathrm{det}(\bm{J})|_{(1,1)}>0$ and $\mathrm{tr}(\bm{J})|_{(1,1)}<0$, and the equilibrium point $(1,1)$ is an ESS point\cite{weibull1997evolutionary}.
\end{enumerate}
\begin{algorithm}[t]
		\caption{Evolutionary Game Solutions for Block Propagation in Public Blockchains}\label{algorithm}

            \textbf{Initialization:} Initialize parameters $\Delta I$, $\Delta U$, $\Delta P$, $\varepsilon$, $R$, and forwarding probability vectors $\bm{x}(t)$ and $\bm{y}(t)$.
            
	    \For{$i \in \bm{x}(t)$}
            {
            \For{$j \in \bm{y}(t)$}
            {
                $t = 1$.

                \While {$\bm{x}(t)$ \rm{and} $\bm{y}(t)$ have not converged \rm{\textbf{and}} $t \leq \text{MAX\_COUNT}$}
                {
                    Use $\mathrm{ODE45}$ \cite{senan2007brief} to solve the replicator dynamics (\ref{H(x)}) (\ref{H(y)}). 

                    $t = t + 1$.
                }
                \textbf{end while}
            }
            \textbf{end for}
            }
            \textbf{end for}
	\end{algorithm}

Motivated by the above analysis, we design an algorithm to solve the evolutionary game for block propagation. In \textbf{Algorithm \ref{algorithm}}, we introduce the strategy evolution of $N$ miners for block propagation in public blockchains. When receiving a new block, miners will choose the best strategy based on the current environment (e.g., the amount of allocated bandwidth and the qualities of channels) to maximize their benefits. Note that the computational complexity of \textbf{Algorithm \ref{algorithm}} is $\mathcal{O}(2T/h)$, where $T = \rm{MAX\_COUNT}$ and $h$ is the time step that is automatically set by the $\rm{ODE45}$ function\cite{senan2007brief}. 

\begin{figure*}[t]
	\centering
        \captionsetup{font=footnotesize}
	\subfloat[ $y(t)>\frac{-\Delta U + \varepsilon R}{\Delta I + \Delta P +\varepsilon R}$ and $(\Delta I + \Delta U) > 0$.]
	{\includegraphics[width=0.3\textwidth]{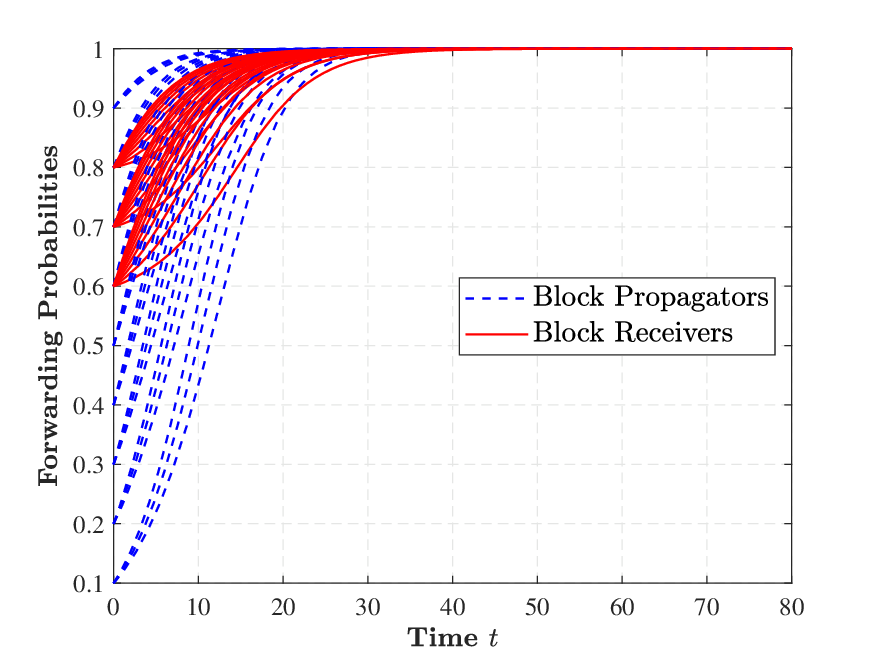}\label{11}}
        \captionsetup{font=footnotesize}
	\subfloat[ $y(t)>\frac{-\Delta U + \varepsilon R}{\Delta I + \Delta P +\varepsilon R}$ and $(\Delta I + \Delta U) < 0$.]
	{\includegraphics[width=0.3\textwidth]{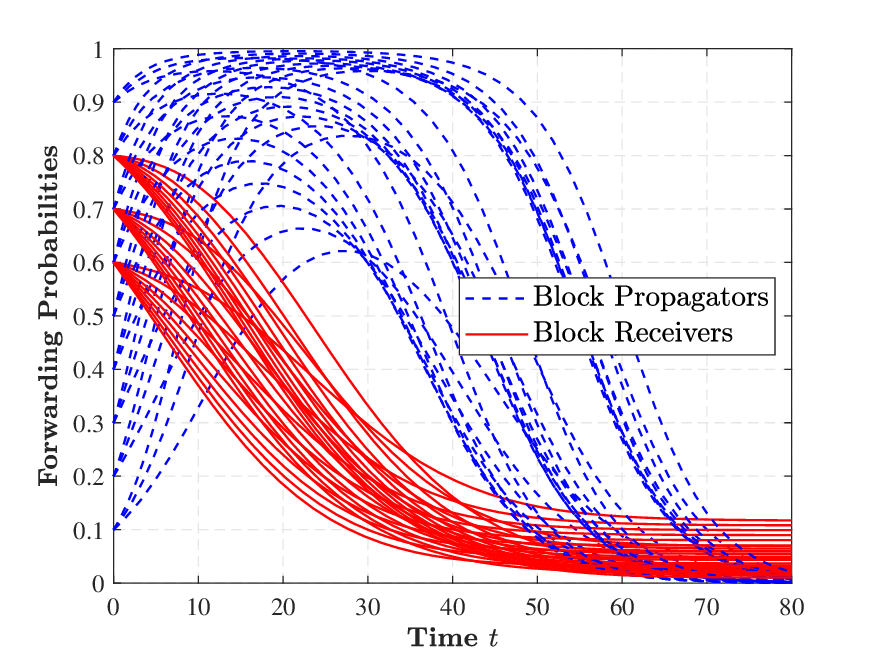}\label{10}}
        \captionsetup{font=footnotesize}
	\subfloat[ $y(t)<\frac{-\Delta U + \varepsilon R}{\Delta I + \Delta P +\varepsilon R}$ and $(\Delta I + \Delta U) < 0$.]
	{\includegraphics[width=0.3\textwidth]{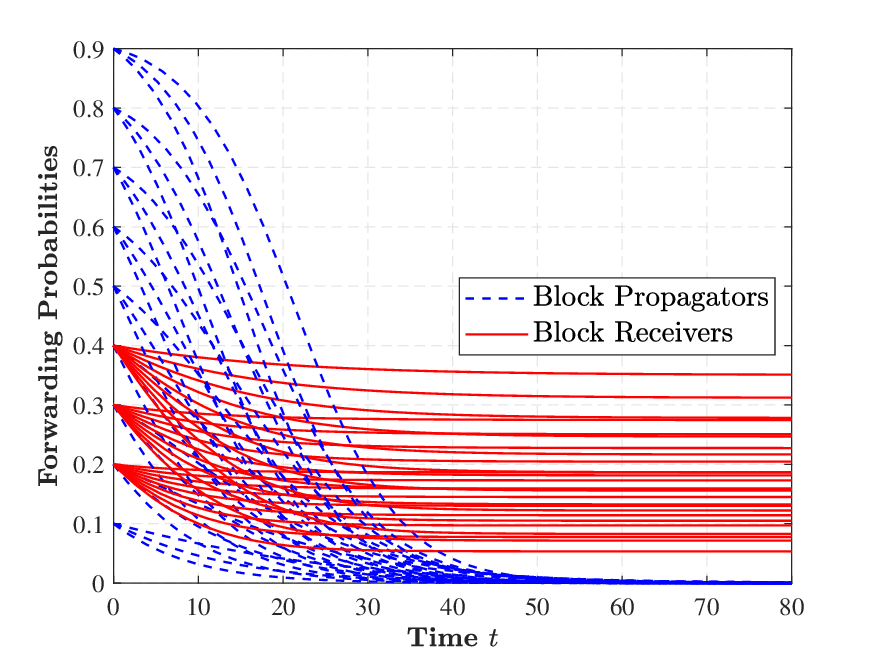}\label{00}}
        \captionsetup{font=footnotesize}
	\caption{Block propagation strategies of miners corresponding to different network conditions.}\label{strategy}
\end{figure*}

\begin{table}[t]\label{parameter}
	\renewcommand{\arraystretch}{1.1} 
        \captionsetup{font=footnotesize}
	\caption{ Key Parameters in the Simulation.}\label{table} \centering 
	\begin{tabular}{m{4.7cm}<{\raggedright}|m{2.7cm}<{\centering}} 
		\hline		
		\textbf{Parameters} & \textbf{Values}\\	
		\hline
		Total number of miners $N$ & $\small\{1000, 2000, 3000, 4000\small\}$ \\	
		\hline
		Number of adjacent miners $k$ & $\small\{2, 3, 4, 5, 6\small\}$  \\	
		\hline
		Cloud-based computing resources of each computing server $C$ & $10^{13}\:\rm{IPS}$  \\
		\hline
		Maximum number of transactions in the
        block $B_{max}$  &  $100$  \\	
		\hline		
		Packing period $T_{p}$ &  $20\:\rm{s}$ \\
		\hline		
		Mining period $T_{mine}$ &  $600\:\rm{s}$\\
		\hline
		Required number of instructions for a transaction to get validated $R_v$ & $10^6$\\
		\hline
		Packing rate of the miner that obtains the bookkeeping right $\tau$ & $[0.5, 5]$\\
		\hline
		Average size of a transaction $P_{size}$ & $300\:\rm{bit}$\\
		\hline
		Effective bit rate per unit bandwidth $R_c$ & $200\:\rm{bps}$\\
		\hline
		Total number of BSs $M$ & $100$\\
		\hline
		Average density of miners that tend to forward blocks $\overline{\omega}$ & $\small\{0.6, 0.7, 0.8, 0.9, 1\small\}$\\
		\hline
		Four probabilities $P_{e}$, $P_{r}$, $P_{f}$, $P_{i}$ & $(0, 1)$\\
		\hline
		Rewards and costs of the evolutionary game $\Delta I$, $I$, $P$, $Q$, $M$, $R$ & $[0, 1]$\\
		\hline
		The unit cost for the punishment risk $\varepsilon$ & $0.1$\\
		\hline
	\end{tabular}\label{parameter}
\end{table}

\section{Numerical Results}\label{result}
{In this section, we verify the effectiveness of the proposed incentive mechanism to provide a theoretical basis for improving block propagation efficiency and analyze factors affecting the minimum value of average AoBI and the miner proportion of different states. We first explore the impacts of the network condition on the block propagation strategies of miners. Then, we compare our proposed scheme with other block propagation mechanisms: i) \textit{Gossip protocol} that miners randomly relay transactions/blocks to their adjacent miners, which is currently used by Bitcoin and Ethereum\cite{hu2022dino}; ii) \textit{Probabilistic flooding approach}\cite{vu2019efficient}. Different from the current gossip protocol implemented by Bitcoin, the probabilistic flooding approach allows miners to maintain certain probabilities of sending information to their adjacent miners based on previous message exchanges between the miners\cite{vu2019efficient}; iii) \textit{Greedy protocol} that miners only consider their current best interests but do not consider the impact of their propagation behaviors on block propagation. Finally, we study the impacts of major factors on the minimum value of average AoBI and block propagation. Note that we use MATLAB to run the experiments on CPU intel i7-8565U and DDR4 8G RAM based on real public blockchain parameters, as shown in Table \ref{parameter} \cite{wenoptimal, rovira2019optimizing, kim2022ensuring, zhao2012sihr, liu2016shir}.}

\begin{figure*}[t]
    \begin{center}
	\begin{minipage}[t]{0.43\linewidth}
		\centering
		\includegraphics[width=1\linewidth]{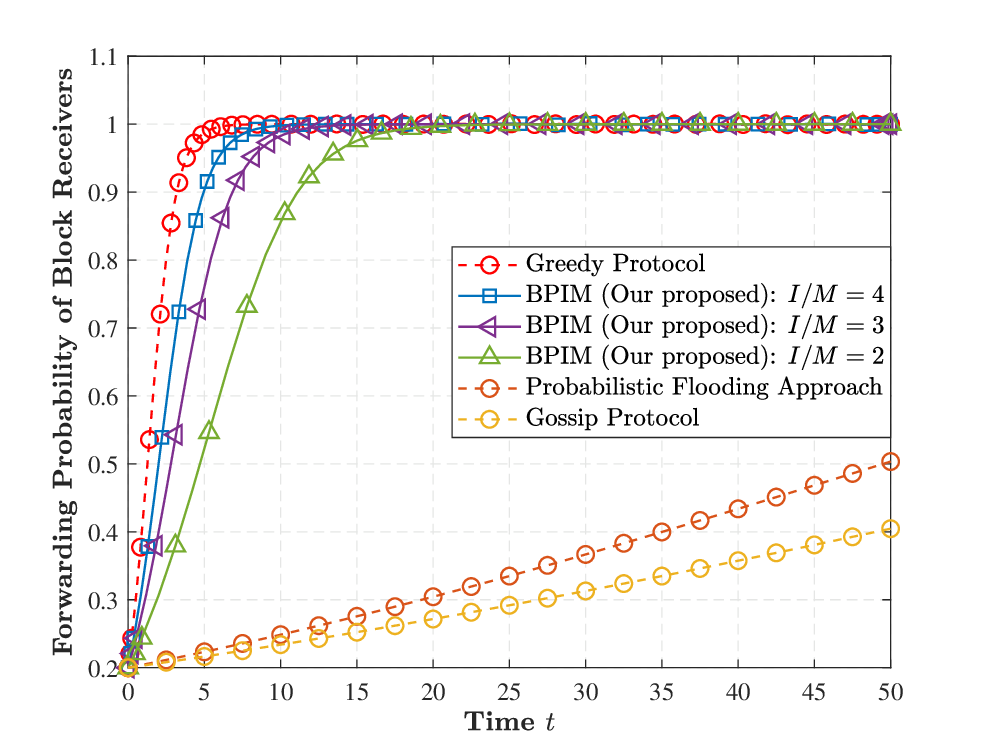}
		\captionsetup{font=footnotesize}
            \caption{The forwarding probability of block receivers for different block propagation mechanisms. Note that the initial value of the block receiver's forwarding probability is set to $0.2$.}\label{BPIM}
	\end{minipage}
	\hspace{0.3in}
	\begin{minipage}[t]{0.43\linewidth}
		\centering
		\captionsetup{font=footnotesize}
		\includegraphics[width=1\linewidth]{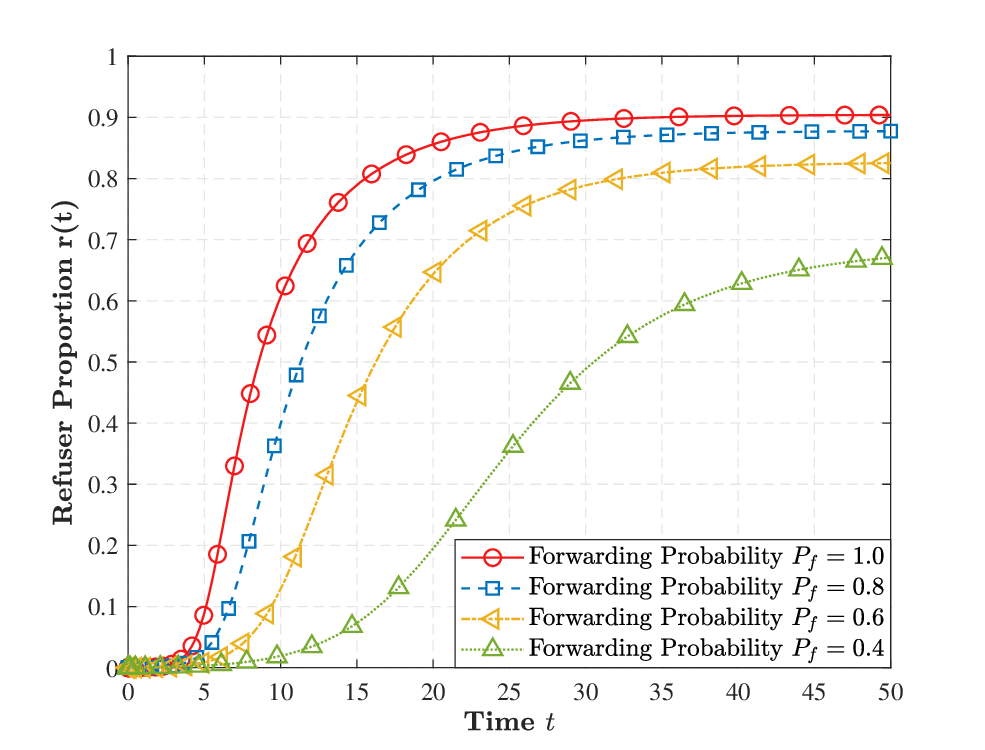}
		\caption{Density of refusers under different forwarding probabilities $P_{f}$ during block propagation, where the total number of miners $N$ is set to $4000$ and the number of adjacent miners $k$ is set to $3$.}\label{rt}
	\end{minipage}

	\begin{minipage}[t]{0.43\linewidth}
		\centering
		\captionsetup{font=footnotesize}
		\includegraphics[width=1\linewidth]{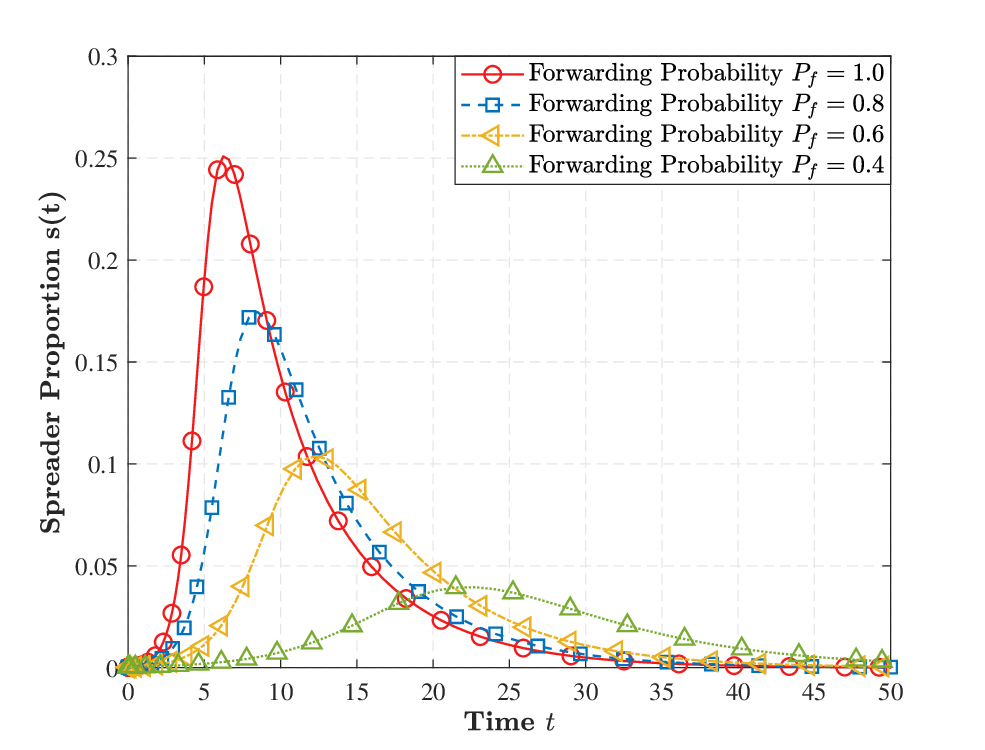}
		\caption{Density of spreaders under different forwarding probabilities $P_{f}$  during block propagation, where the total number of miners $N$ is set to $4000$ and the number of adjacent miners $k$ is set to $3$.}\label{st}
	\end{minipage}
	\hspace{0.3in}
	\begin{minipage}[t]{0.43\linewidth}
		\centering
		\captionsetup{font=footnotesize}
		\includegraphics[width=1\linewidth]{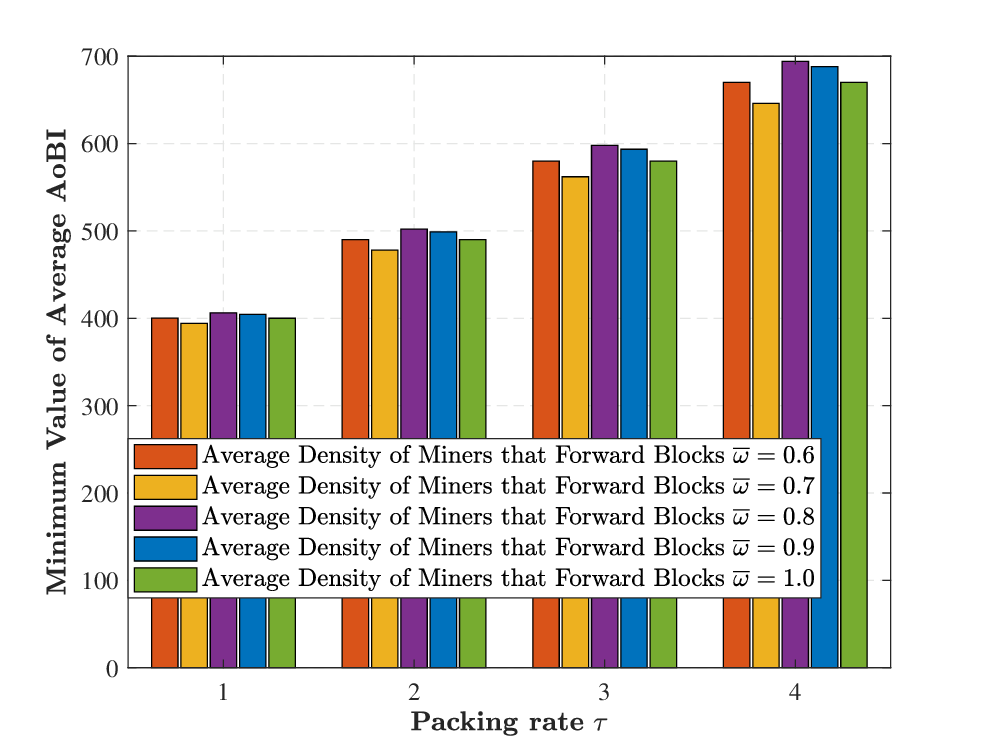}
		\caption{The minimum value of average AoBI corresponding to different packing rates $\tau$, where the total number of miners $N$ and the number of adjacent miners $k$ are set to $4000$ and $3$, respectively.}\label{AoBI_w}
	\end{minipage}
	\end{center}
\end{figure*}

\subsection{Impacts of Network Conditions on Block Propagation Strategies of Miners}
First, we analyze the impacts of the network condition on the block propagation strategies of miners, as shown in Fig. \ref{strategy}. {From Fig. \ref{11}, we can find that when $y(t)>\frac{-\Delta U + \varepsilon R}{\Delta I + \Delta P +\varepsilon R}$ and $(\Delta I + \Delta U) > 0$, the equilibrium point is $(1,1)$ and is an ESS point, indicating that as time progresses, both block propagators and block receivers forward the block, which helps reach the whole network consensus faster. From Fig. \ref{10}, we can find that when $y(t)>\frac{-\Delta U + \varepsilon R}{\Delta I + \Delta P +\varepsilon R}$ and $(\Delta I + \Delta U) < 0$, block receivers tend to not forward the block, while block propagators tend to forward the block initially and then not forward the block. The reason is that since block receivers do not forward the block, block propagators have to undertake the punishment risk. From Fig. \ref{00}, we find that both block propagators and receivers tend not to forward the block. To sum up, when $\Delta U > \varepsilon R $, namely the basic propagation reward is greater than the cost that spreaders forward the new block to evildoers, the equilibrium state of both block propagators and receivers is forwarding the block.} Therefore, the whole network consensus can be reached more quickly.

\subsection{Efficiency of the Proposed Incentive Mechanism}
\begin{figure*}[t]
    \begin{center}
        \begin{minipage}[t]{0.43\linewidth}
		\centering
		\captionsetup{font=footnotesize}
		\includegraphics[width=1\linewidth]{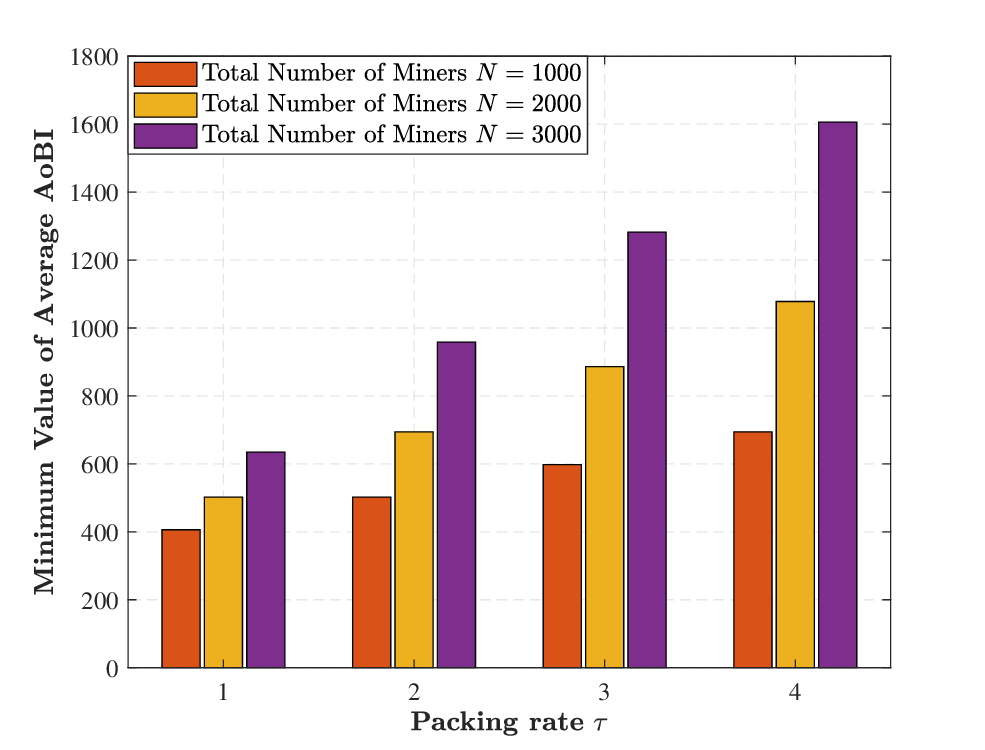}
		\caption{The minimum value of average AoBI corresponding to different packing rates $\tau$, where the number of adjacent miners $k$ and the average density of miners that forward blocks $\overline{\omega}$ are set to $3$ and $0.8$, respectively.}\label{AoBI_N}
	\end{minipage}
        \hspace{0.3in}
	\begin{minipage}[t]{0.43\linewidth}
		\centering
		\captionsetup{font=footnotesize}
		\includegraphics[width=1\linewidth]{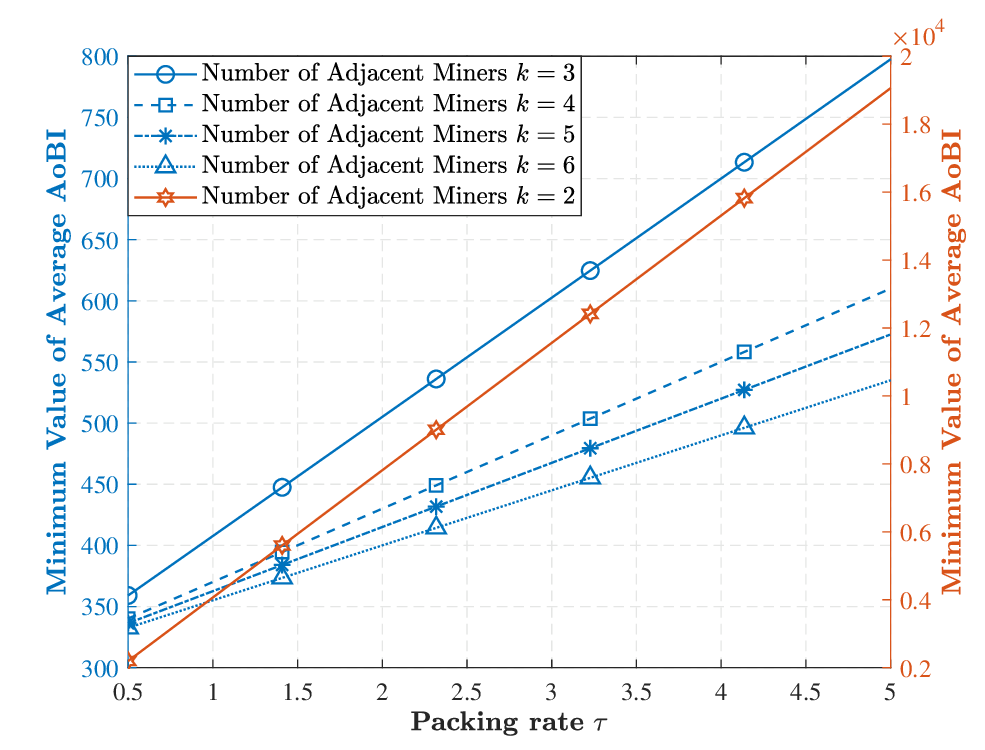}
		\caption{The minimum value of average AoBI corresponding to different packing rates $\tau$, where the total number of miners and the average density of miners that forward blocks $\overline{\omega}$ are set to $1000$ and $0.5$, respectively.}\label{AoBI_k}
	\end{minipage} 
	\hspace{0.3in}
	\begin{minipage}[t]{0.43\linewidth}
		\centering
		\captionsetup{font=footnotesize}
		\includegraphics[width=1\linewidth]{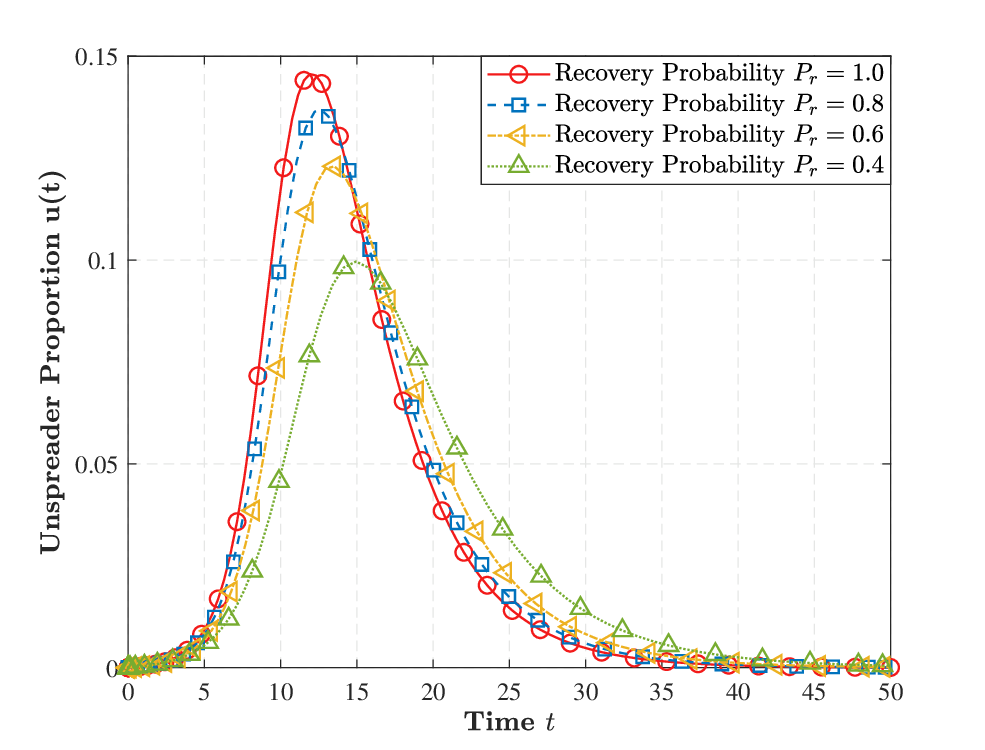}
		\caption{Density of unspreaders under different recovery probabilities $P_{r}$ during block propagation, where the forwarding probability $P_f$ is set to $0.5$.}\label{ut}
	\end{minipage}
	\hspace{0.3in}
	\begin{minipage}[t]{0.43\linewidth}
		\centering
		\captionsetup{font=footnotesize}
		\includegraphics[width=1\linewidth]{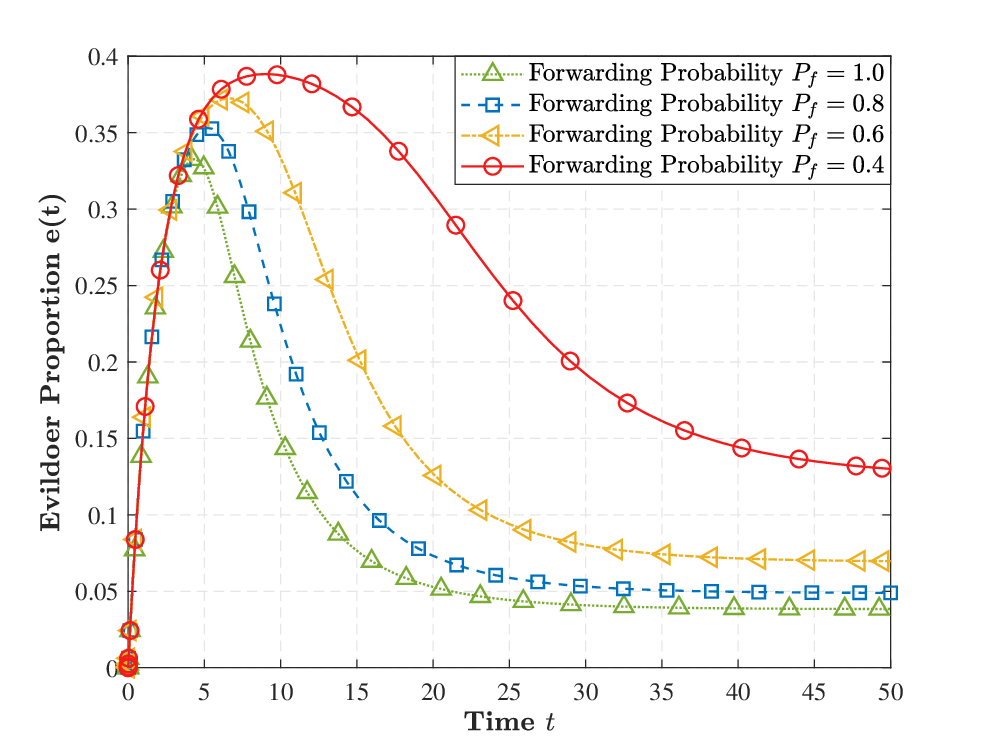}
		\caption{Density of evildoers under different forwarding probabilities $P_{f}$ during block propagation, where the total number of miners $N$ is set to $4000$ and the number of adjacent miners $k$ is set to $3$.}\label{et}
	\end{minipage}
	\end{center}
\end{figure*}
To evaluate the performance of the proposed scheme over other block propagation mechanisms, we show the forwarding probability of block receivers under different block propagation mechanisms in Fig. \ref{BPIM}. We can observe that the forwarding probability of block receivers corresponding to each block propagation mechanism increases as time progresses, which indicates that as the interaction continues, more miners approve the new block and forward it to their adjacent miners rather than discard it. Besides the greedy protocol that performs best as expected, our proposed scheme is always ahead of the probabilistic flooding approach and the gossip protocol in the forwarding probability of block receivers. The reason is that with the role of the incentive, resource-limited miners prefer to forward blocks to suitable adjacent miners that tend to forward the block, thereby obtaining more benefits. Moreover, the greater the incentive strength, the better the performance of the proposed incentive mechanism. For example, the green line (i.e., incentive strength $I/M=2$) means that it takes about $20$ epochs for forwarding probability to reach the upper bound, while the blue line (i.e., incentive strength $I/M=4$) means that it only takes about $10$ epochs for forwarding probability to reach the upper bound.

To verify that increasing the forwarding probability can achieve the whole network consensus faster, we show the density of refusers changing over time for different forwarding probabilities $P_{f}$ in Fig. \ref{rt}. We can find that the density of refusers corresponding to different $P_{f}$ increases with time. When other probabilities are fixed, the larger the forwarding probability, the faster the number of refusers grows, indicating that the higher the forwarding probability, the more miners complete the block validation, which reaches block consensus faster. Therefore, by combining Fig. \ref{BPIM} and Fig. \ref{rt}, we can conclude that the performance of the proposed incentive mechanism is better than those of the probabilistic flooding approach and the gossip protocol. The reason is that under the role of the incentive mechanism, block propagators consider their own interests and forward blocks to appropriate adjacent miners that have not received blocks and tend to forward blocks, which is not only conducive to improving block propagation efficiency but also avoids the block retransmission problem. Moreover, miners under the greedy protocol only consider the local optimality of block propagation, which is bound to greatly increase the redundancy of the block in the miner network. \textit{In summary, the overall performance of the proposed incentive mechanism is better than those of current block propagation mechanisms}.

Figure \ref{st} shows the density of spreaders changing over time for different forwarding probabilities $P_{f}$. As we expect, when other probabilities are fixed, the larger the forwarding probability, the faster the number of spreaders grows, indicating the larger the average density of adjacent miners that tend to forward blocks $\overline{\omega}$. Figure \ref{AoBI_w} illustrates the minimum value of average AoBI corresponding to different packing rates $\tau$ under different $\overline{\omega}$. From Fig. \ref{AoBI_w}, we can find that when $B_{max}\sqrt{\frac{R_v M R_c W}{C P_{size}\overline{\omega}}}<2$, the minimum value of average AoBI increases as the packing rate $\tau$ increases. The proof can be seen in Appendix (B). In addition, for the fixed packing rate, the minimum value of average AoBI first decreases and then increases as $\overline{\omega}$ increases, which indicates that there exists an optimal $\overline{\omega}$ to minimize the average AoBI of the miner network. The reason is that with the increase in the number of adjacent miners that forward blocks, the minimum value of average AoBI decreases. However, the more the number of adjacent miners that forward blocks, the less bandwidth resources are allocated to them for miner communications, thereby increasing the minimum value of average AoBI.

\subsection{Impacts of Factors on the Minimum Value of Average AoBI and Block Propagation}
 Figure \ref{AoBI_N} illustrates the minimum value of average AoBI corresponding to different packing rates $\tau$ under different total numbers of miners $N$. As we thought, the larger the total number of miners, the larger the minimum value of average AoBI, namely the larger the average AoBI of the miner network as well. The reason is that the increase in the total number of miners means that it takes more time to complete network-wide block validation and block propagation, making the whole network consensus latency larger. Figure \ref{AoBI_k} illustrates the minimum value of average AoBI corresponding to different packing rates $\tau$ under different numbers of adjacent miners $k$. Since the value of $(\overline{\omega}k)$ exists in both cases, we present them in Fig. \ref{AoBI_k} to better analyze the effect of $k$ on the minimum value of average AoBI. The orange line corresponds to the case of $\overline{\omega}k = 1$ and the other lines correspond to the case of $\overline{\omega}k \neq 1$. From Fig. \ref{AoBI_k}, we can observe that the minimum value of average AoBI corresponding to the case of $\overline{\omega}k = 1$ is much higher than that corresponding to the case of $\overline{\omega}k \neq 1$. For the case of $\overline{\omega}k \neq 1$, the smaller $k$, the larger the minimum value of average AoBI. The reason is that a smaller $k$ means that fewer adjacent miners forward the block, which reduces block propagation efficiency and makes the overall block propagation latency of the miner network exponentially larger, thus increasing the average AoBI of the network.

Figure \ref{ut} shows the density of unspreaders changing over time for different recovery probabilities $P_{r}$. From Fig. \ref{ut}, we can obverse that when other probabilities are fixed, the larger the recovery probability, the greater the growth rate and decline rate of unspreaders, and the number of unspreaders is larger. The reason is that more ignorants participate in validating the block as recovery probability increases, leading to an increase in the number of unspreaders. After completing the block validation, unspreaders are converted into refusers with the probability $P_i$. Therefore, the larger the number of unspreaders, the faster the decline of unspreaders. Figure \ref{et} shows the density of evildoers changing over time for different forwarding probabilities $P_{f}$. From Fig. \ref{et}, we can see that the larger the forwarding probability, the smaller the number of evildoers, and the greater the decline rate of evildoers. {The reason is that the higher forwarding probability indicates that miners that forward blocks can obtain more benefits than those undertaking punishment risk by doing evil. Therefore, the proposed incentive mechanism can not only improve block propagation efficiency but also decrease the likelihood of evil behavior by miners.}

\section{Conclusion and Future Work}\label{conclude}
{In this paper, we focused on improving the performance of blockchain-enabled Web 3.0, especially optimizing block propagation. Specifically, we proposed a novel freshness metric called AoBI based on the concept of AoI for public blockchains to measure block freshness. To make block propagation optimization tractable, we classified miners into five different states and proposed a block propagation model for public blockchains inspired by epidemic models. To achieve block propagation optimization, we then established an incentive mechanism based on the evolutionary game for improving block propagation efficiency. Finally, numerical results demonstrate that compared with other block propagation mechanisms, the proposed incentive mechanism can achieve block propagation optimization and decrease the minimum value of average AoBI, in which the greater the incentive strength, the higher block propagation efficiency.}

{For future work, considering the impact of the energy consumption of miners on block propagation delay, we will further optimize AoBI to better measure block freshness. Besides, we will systematically explore the potential impact of the proposed incentive mechanism on the decentralization and security of blockchain-enabled Web 3.0. Moreover, we can design a prototype system to evaluate our scheme and use AI tools such as DRL to solve the evolutionary game, which can objectively reflect the interaction process between miners for block propagation.}

\section*{Appendix A}\label{appendix_A}
We consider it a normal case that miners that have forwarded the new block will never forward it even if receiving it again. Considering that reaching the whole network consensus needs $(m+1)$ rounds of validating and propagating the block, we can obtain
\begin{equation}\label{wk}
    \begin{aligned}
        k + \overline{\omega}k^2 + \overline{\omega}^2k^3+ \cdots + \overline{\omega}^mk^{m+1} \geq N,\: \overline{\omega}k \geq 1.
    \end{aligned}
\end{equation}
\begin{itemize}
    \item When $\overline{\omega}k = 1$, we can rewrite
    (\ref{wk}) as
    \begin{equation}\label{k_N}
        \begin{aligned}
            \underbrace{k + k + k + \cdots + k}_{m+1} \geq N.\\
        \end{aligned}
    \end{equation}
    Based on (\ref{k_N}), the value range of $(m+1)$ is given by
    \begin{equation}
        \begin{aligned}
            m+1\geq \frac{N}{k}.
        \end{aligned}
    \end{equation}
    Due to $(m+1) \in \mathbb{Z}^+$, we can obtain
    \begin{equation}
        \begin{aligned}
            m + 1 = \bigg\lceil \frac{N}{k} \bigg\rceil.
        \end{aligned}
    \end{equation}
    \item When $\overline{\omega}k > 1$, we can rewrite (\ref{wk}) as
    \begin{equation}
        \begin{aligned}
            \frac{k[1-(\overline{\omega}k)^{m+1}]}{1-\overline{\omega}k} \geq N,
        \end{aligned}
    \end{equation}
    that is 
        \begin{equation}\label{wk^m}
            \begin{aligned}
              (\overline{\omega}k)^{m+1} \geq \frac{N(\overline{\omega}k-1)+k}{k} > 0.
            \end{aligned}
        \end{equation}
        Taking the logarithm with base $(\overline{\omega}k)$ on both sides of (\ref{wk^m}) simultaneously, we can obtain
        \begin{equation}
            \begin{aligned}
              m+1 \geq \log_{\overline{\omega}k}\bigg(\frac{N(\overline{\omega}k-1)+k}{k}\bigg).
            \end{aligned}
        \end{equation}
        Due to $(m+1) \in \mathbb{Z}^+$, we can obtain
        \begin{equation}
            \begin{aligned}
              m+1 = \Bigg\lceil \log_{\overline{\omega}k}\bigg(\frac{N(\overline{\omega}k-1)+k}{k}\bigg) \Bigg\rceil.
            \end{aligned}
        \end{equation}
\end{itemize}

In summary, the number of rounds for reaching the whole network consensus is given by
\begin{equation}
    m + 1 = \left\{
    \begin{aligned}
      &\: \Bigg\lceil \log_{\overline{\omega}k}\bigg(\frac{N(\overline{\omega}k-1)+k}{k}\bigg) \Bigg\rceil, \quad \overline{\omega}k > 1,\\
      &\: \bigg\lceil \frac{N}{k} \bigg\rceil, \quad \overline{\omega}k = 1.
    \end{aligned}
    \right.
\end{equation}
Therefore, \textbf{Proposition \ref{P_1}} and \textbf{Proposition \ref{P_2}} are proved.

\section*{Appendix B}\label{appendix_B}
As shown above, the minimum value of average AoBI is essentially a function $\varphi(\tau) = a\tau + \frac{b}{\tau}, \tau \in [\frac{1}{T_p}, \frac{B_{max}}{T_p}]$, where $a = \frac{P_{size}T_p\overline{\omega}N}{M R_c W}, b = \frac{R_v N B_{max}^2}{4CT_p}, a, b > 0$. Taking the derivative of $\varphi(\tau)$, we have
\begin{equation}
    \varphi^\prime(\tau) = \frac{a\tau^2-b}{\tau^2}.
\end{equation}
To prove that the minimum value of average AoBI increases monotonically with the increase of $\tau$ on $[\frac{1}{T_p}, \frac{B_{max}}{T_p}]$, $\varphi^\prime(\tau)$ should be a constant that is greater than $0$ on $[\frac{1}{T_p}, \frac{B_{max}}{T_p}]$. Thus, we have $\frac{1}{T_p} > \sqrt{\frac{b}{a}}$, namely
\begin{equation}
    \begin{aligned}
      \frac{1}{T_p} > \frac{B_{max}}{2T_p}\sqrt{\frac{R_v M R_c W}{C P_{size} \overline{\omega}}},
    \end{aligned}
\end{equation}
that is 
\begin{equation}
    \begin{aligned}
      0< B_{max}\sqrt{\frac{R_v M R_c W}{C P_{size}\overline{\omega}}} < 2.
    \end{aligned}
\end{equation}
Therefore, the proof is completed.

\bibliographystyle{IEEEtran}
\bibliography{ref}
\end{document}